\documentclass[12pt,hidelinks]{article}
\usepackage{graphicx}
\usepackage{float}
\usepackage{enumerate}
\usepackage{natbib}
\usepackage{bibunits}          
\defaultbibliographystyle{apalike} 
\defaultbibliography{CausalMissing.bib}
\usepackage{colortbl}
\usepackage{color}
\usepackage[dvipsnames]{xcolor}
\usepackage{easyReview}

\usepackage{amsmath,amssymb,amsthm,authblk,booktabs,bm,geometry,multirow,threeparttable,multicol,ragged2e,tikz}  
\usetikzlibrary{arrows.meta}
\usepackage{derivative,esdiff}
\allowdisplaybreaks
\usepackage{hyperref,url}
\hypersetup{
    colorlinks=false,
    linkcolor=blue,
    urlcolor=blue,
    citecolor=blue
}

\newtheorem{theorem}{Theorem}
\newtheorem{proposition}{Proposition}
\newtheorem{assumption}{Assumption}
\newtheorem{example}{Example}

\newtheorem{lemma}{Lemma}

\def\T{{ \mathrm{\scriptscriptstyle T} }}

\newcommand*{\dif}{\mathop{}\!\mathrm{d}}

\makeatletter
\newcommand*{\ind}{%
	\mathbin{%
		\mathpalette{\@ind}{}%
	}%
}
\newcommand*{\nind}{%
	\mathbin{% % The final symbol is a binary math operator
		\mathpalette{\@ind}{\not}% \mathpalette helps for the adaptation
	}%
}
\newcommand*{\@ind}[2]{%
	\sbox0{$#1\perp\m@th$}% box 0 contains \perp symbol
	\sbox2{$#1=$}% box 2 for the height of =
	\sbox4{$#1\vcenter{}$}% box 4 for the height of the math axis
	\rlap{\copy0}% first \perp
	\dimen@=\dimexpr\ht2-\ht4-.2pt\relax
	\kern\dimen@
	{#2}%
	\kern\dimen@
	\copy0 % second \perp
} 
\makeatother

\def\H{{\mathbb H}}
\def\spacingset#1{\renewcommand{\baselinestretch}%
{#1}\small\normalsize} \spacingset{1}

\title{\bf A generalized tetrad constraint for testing
conditional independence given a latent
variable}

\author[1]{Naiwen Ying}
\author[1]{Ping Zhang}
\author[2]{Shanshan Luo}
\author[1]{Wang Miao}
  
\affil[1]{Department of Probability and Statistics, Peking University}
\affil[2]{Department of Applied Statistics, Beijing Technology and Business University}
\date{}

\begin{document}
\maketitle
\begin{abstract}
The  tetrad   constraint is widely used to test whether four observed variables are conditionally independent given a latent variable, 
based on the fact that if  four observed variables following a linear model are mutually independent after conditioning on an unobserved variable, 
then  products of covariances of any two different pairs of these four variables are equal.
It is an important tool for discovering a latent common cause or distinguishing  between alternative linear causal structures.
However, the classical tetrad constraint fails in nonlinear models because the covariance of observed variables cannot capture nonlinear association. 
In this paper, we propose a generalized tetrad constraint,
which establishes a testable implication for  conditional independence given a latent variable in nonlinear and nonparametric models.
In   linear models, this   constraint implies the  classical tetrad constraint;
in nonlinear models, it remains a necessary condition for  conditional independence but the classical tetrad constraint no longer is.
Based on this constraint, we further propose a formal test, 
which can control  type I error
and has power approaching unity under certain conditions.
We illustrate the proposed approach via simulations and two real data applications on mental ability tests and on moral
attitudes towards dishonesty.
\end{abstract}

\noindent%
\vfill

\spacingset{1.5} 

\begin{bibunit}
\section{Introduction}

The tetrad constraint is widely used to test whether four correlated  variables are conditionally independent given 
a latent variable.
Suppose  one has observed four variables,
let $\sigma_{ij}, i,j\in\{1,2,3,4\}$ denote the covariance between any two of them. 
The following quantities are called tetrad differences,
\[\sigma_{12}\sigma_{34}-\sigma_{13}\sigma_{24}, \quad \sigma_{12}\sigma_{34}-\sigma_{14}\sigma_{23}, \quad \sigma_{13}\sigma_{24}-\sigma_{14}\sigma_{23},\]
which are the differences between products of covariances of any two different pairs of these four variables.
All three  tetrad differences are zero if the observed variables are mutually independent after conditioning on a latent variable and  their conditional expectations follow a linear model, which  is known as the tetrad constraint.
Based on this result, one can test whether   the observed correlation is  due to a latent  common cause.
This idea can date back to \citet{spearman1904general,spearman1927abilities}, 
who argued that all manifestations of intelligence consist of two types of factors: the general factor that is a common component of all forms of intellectual activities, and the specific factors that vary with different intellectual activities.
There exist a variety of statistical methods for testing the tetrad constraint,
for example, \citet{wishart1928sampling}  proposed  a testing method  based on the asymptotic normality of  sample tetrad difference  under  a multivariate normal model for observed variables. 
\citet{bollen1990outlier} extended the asymptotic test to possibly non-Gaussian variables and established simultaneous test for multiple tetrad differences in the presence of more than four observed variables.
Such tetrad tests have been  implemented in empirical studies to confirm or reject a linear structural equation model of interest based on  the sample covariance matrix  of observed variables \citep{bollen1993confirmatory,bollen2000tetrad}.

In recent years, the tetrad constraint becomes increasingly popular in causal inference for the  detection of  latent causes.
Particularly in causal discovery with latent variables,
researchers examine the vanishing tetrad differences to  discover causal relations from observational data and distinguish between alternative causal structures in linear structural equation
 models in the presence of latent variables \citep{glymour1987discovering,shafer1993generalization,spirtes2000causation}.
Furthermore,
\citet{sullivant2010trek} and \citet{spirtes2013calculation} generalize the tetrad constraint  and propose a trek separation theorem,
which establishes  the equivalence of  rank constraints on the covariance matrix of arbitrary numbers of observed variables and the graphical separation criteria.
Based   on  the tetrad constraint and   trek separation theorem,  
researchers  have developed  algorithms to identify clusters of observed variables that share the same latent cause and the relationship between 
different latent variables \citep{silva2006learning,kummerfeld2016causal}, when there exists a partition of observed variables defined by the latent causes and there are no direct relations between observed variables, called a pure measurement model.
\citet{cai2019triad,xie2020generalized} propose the triad constraint
and the generalized independent noise (GIN) condition
for discovering  the structure of latent variables in linear non-Gaussian 
causal models.
Tetrad constraint, trek separation theorem and the GIN condition have also been applied in  instrumental variable (IV) selection and other causal inference  problems \citep{kuroki2005instrumental,silva2017learning,xie2022testability}.
These previous causal discovery algorithms depend  on linear structural models, 
which  may be overly restrictive in practice. 
There are examples that the classical tetrad constraint fails to discover conditional independence even with a slight violation of linearity.
To the best of our knowledge,  tetrad constraint for nonlinear models has not been well established yet.
The classical tetrad constraint depends crucially on the connection between the covariance matrix of observed variables and regression coefficients of linear structural equations, 
which is no longer a suitable characterization of dependence between variables in nonlinear models.

In this paper we propose a tetrad constraint under  a more flexible assumption that allows for  nonlinear models,
which generalizes the classical   tetrad constraint. 
The key is to employ the confounding bridge function to capture the nonlinear association between observed variables, in contrast to using covariance matrix.
The confounding bridge function has been used in causal inference  to characterize the relationship between confounding effects on different outcome variables \citep{miao2018identifying,cui2024semiparametric}, by viewing the outcome variables as proxies of the latent variable. 
Consider testing conditional independence of four variables $(X,Y,Z,W)$ given a latent factor $U$.
In the same spirit, we define for observed variable $Y$ the confounding bridge functions  as  transformations $g(Z), h(W)$ of variables $Z$ and $W$, respectively so that the effects of $U$ on $Y,g(Z),h(W)$ are equal.
We show that under certain conditions, if the null hypothesis is correct, 
then these two transformations and $Y$ have   equal  expectation  conditional on variable $X$,
which we call the generalized tetrad constraint.
The differences between these conditional expectations 
quantify the degree of departure from the null hypothesis,
which we call the generalized tetrad differences.
The   generalized tetrad constraint implies the classical tetrad
constraint if the confounding bridge functions are linear,
which holds for example when the four observed variables follow a   joint normal distribution.
For nonlinear models, violation of the generalized tetrad constraint indicates that the null hypothesis is incorrect; however, there are  examples that the classical tetrad constraint may be falsely  satisfied in this case.

Based on the generalized tetrad constraint, we propose a formal procedure
for testing the null hypothesis.
Note that  the generalized tetrad differences  are differences in conditional expectations, which are complicated functionals.
We construct a one-dimensional measure for each generalized tetrad difference,
which is equal to zero if and only if 
the corresponding generalized tetrad difference is zero.
Our measure of generalized tetrad difference is motivated by the martingale difference divergence \citep[MDD,][]{shao2014martingale} designed for characterizing the conditional mean dependence between two random variables,
which is the nonnegative square root of a weighted integral of squared correlations between {the unconditional variable and the Fourier bases of the conditional variable.} 
Our  measure is different from the MDD in that the confounding bridge functions in the expectation are unknown and need to be estimated.
We aggregate the measure for two generalized tetrad differences together to assess the deviation from the generalized tetrad constraint, called the aggregated measure of generalized tetrad differences (AMGT).
For testing the null hypothesis,
We first obtain estimators of the confounding bridge functions,
and then construct an estimator of the AMGT.
The asymptotic distribution of the AMGT estimator after standardization follows a weighted sum of independent Chi-square variables with one degree of freedom,
serving as the basis for testing the null hypothesis. 
We show that the proposed test can control type I error and has power approaching unity under certain conditions.

The remainder of this paper is organized as follows. 
In Section~\ref{sec:challenge}, we briefly review the challenge for the classical tetrad constraint in nonlinear models.
In Section~\ref{sec:intro}, we introduce a generalized tetrad constraint and show its theoretical properties and connections to the classical tetrad constraint. 
We  also introduce the AMGT and illustrate its relationship with the generalized tetrad constraint. 
In Section~\ref{sec:tetrad-test}, we develop a test statistic based on the generalized tetrad constraint and the AMGT.
We illustrate how to estimate the nuisance models for constructing the test statistic, and obtain its asymptotic distribution.
In Section~\ref{sec:simulation}, we study the finite-sample performance of our proposed approach and the classical tetrad test via simulations,
and in Section~\ref{sec:application}, we apply the proposed and competing methods to two real data examples about  mental ability tests and moral attitudes towards dishonesty, respectively. 
Our analysis of  Holzinger and Swineford’s (HS) data set about mental ability tests supports the existence of  a latent variable  as  a common factor of four scores of spatial tests, which  can be interpreted as ``spatial ability''.
Our analysis of the World Values Survey data set supports the existence of  a latent variable as a common factor of the attitudes towards four dishonest behaviors, which can be interpreted as honesty. 
These results are consistent with existing theories about mental ability and honesty.
Section~\ref{sec:discussion} concludes with a brief discussion on possible extensions of the proposed approach.

Throughout the paper,  we use ``i.i.d.'' for ``independent and identically distributed''. 
Let $f$ denote a generic probability density or mass function.
Let $\sigma_{XY}$ denote  the covariance between  two random variables $X$ and $Y$. 
Let $i$ denote the imaginary unit.
For a complex number $z$, let $\overline{z}$ denote the complex conjugate, and $\vert z\vert^2=z\overline{z}$. 
Denote the Euclidean norm of a vector $\boldsymbol{x}$ as $\Vert \boldsymbol{x}\Vert.$
For a complex-valued function $g=g(x)$, let $\Vert g\Vert_{\infty}=\sup_{x}\vert g(x)\vert$  denote the supremum norm. 
The Gateaux derivative of a functional $\phi$ at $g_0$ in the direction $v$ is denoted as  
$\odif\phi(g_0)/\odif g[v]=\partial\phi(g_0+tv)/\partial t\big|_{t=0}.$
Vectors are assumed to be column vectors, unless explicitly
transposed.

\section{Challenge for the classical tetrad constraint in nonlinear models}\label{sec:challenge}

Suppose four variables $(X,Y,Z,W)$ are observed, which are  confounded by a common unobserved variable $U$.
We are interested in whether these four observed variables are conditionally independent given  $U$, that is, 
\[\mathbb H_0: X\ind Y\ind Z\ind W\mid U.\]
Figure~\ref{fig:1} provides a graphical illustration for $\mathbb H_0$.  
Although there exists a huge literature on testing conditional independence between observed variables \citep[for example,][]{su2007consistent,wang2015conditional,ai2024testing}, 
the key difference and difficulty here is that the confounder $U$ is not observed.
This problem originates from \citet{spearman1904general} in his theory about   the existence of  a general factor of human intelligence.
Considering four branches of intellectual activities: Classics, French, English, and Mathematics,
the correlation of each two variables, based on the school examination gradings,  revealed strong associations between different branches of intellectual activities after adjustment for observed covariates such as gender and age.
Spearman argued that there exists a single common factor underlying different activities, called general intelligence, 
and the remaining variability in each variable can be explained with different specific factors due to domain-specific skills, 
which is known as Spearman's two-factor theory of intelligence.   
The null hypothesis $\H_0$ is such a model describing the existence of  a general factor of human intelligence.
Besides, testing conditional independence given a latent variable is also an important problem in economics, social science and causal discovery.

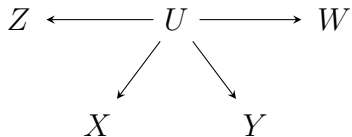
\begin{figure}[tb]
	\centering
	 		\begin{tikzpicture}[scale=0.7,
	 		->,
	 		shorten >=2pt,
	 		>=stealth,
	 		node distance=1cm,
	 		pil/.style={
	 			->,
	 			thick,
	 			shorten =2pt,}
	 		]
	 		\node (X) at (1.5,-2) {$X$};
	 		\node (U) at (3,0) {$U$};
	 		\node (W) at (6,0) {$W$};
	 		\node (Y) at (4.5,-2) {$Y$};
	 		\node (Z) at (0,0) {$Z$};     								
	 		\foreach \from/\to in {U/Y,U/X,U/Z,U/W}
	 		\draw (\from) -- (\to);
	 		\end{tikzpicture}
    \caption{Causal diagram for $\H_0$.}
    \label{fig:1}
	\end{figure}

In order to test $\H_0$, a widely-used method is the tetrad constraint for the linear structural equation model. 
Suppose 
\begin{equation}
\label{eq:1}
\begin{gathered}
X=\alpha_0+\alpha_1U+\varepsilon_X,\quad Y=\beta_0+\beta_1U+ \varepsilon_Y,\\
Z=\gamma_0+\gamma_1U+\varepsilon_Z,\quad W=\eta_0+\eta_1U +\varepsilon_W,\\
(\varepsilon_X,\varepsilon_Y,  \varepsilon_Z, \varepsilon_W,  U) \text{ has diagonal covariance matrix.}    
\end{gathered}
\end{equation}
This model   is sometimes called an effect indicator model \citep{bollen2000tetrad}, 
where a latent factor is of interest, 
these four observed variables are  viewed as indicators and measurements of the latent factor, 
and they are correlated only through the latent factor.
Under model \eqref{eq:1},
we have
\[\sigma_{XY}\sigma_{ZW}=\sigma_{XZ}\sigma_{YW}=\sigma_{XW}\sigma_{YZ},\]
which is known as the  tetrad constraint \citep{spearman1927abilities,bollen1993confirmatory,bollen2000tetrad,spirtes2000causation}.
The tetrad constraint holds for all possible values of the coefficients and all possible
distributions of the exogenous variables $(\varepsilon_X,\varepsilon_Y,  \varepsilon_Z, \varepsilon_W,  U)$ in model \eqref{eq:1}, which is also called entailed by model \eqref{eq:1}, in the sense of \citet{spirtes2013calculation}.
The tetrad constraint only depends on the observed variables and thus offers a testable implication of the effect indicator model.
Identifying an effect indicator model is an important issue in confirmatory tetrad analysis \citep{bollen1993confirmatory,bollen2000tetrad}, where
models are specified according to specialized knowledge, and the structure of each model often implies population tetrads that should be zero.
One can apply simultaneous testing of multiple tetrad differences to examine whether a model is consistent with the data.

However, this classical tetrad constraint relies crucially on the linear model,
which is sometimes overly restrictive in practice. 
Under nonlinear models, this classical tetrad constraint may fail and lead to  biased inference on $\H_0$. 
Here is such a counterexample.
\begin{example}\label{ex:1}
Assuming the following   model 
\begin{equation}\label{eq:2}
\begin{gathered}
X=\alpha_0+\alpha_1U+\alpha_2U^2+\varepsilon_X,\ Y=\beta_0+\beta_1U+\beta_2U^2+\varepsilon_Y,\\
Z=\gamma_0+\gamma_1U+\varepsilon_Z,\ W=\eta_0+\eta_1U+\varepsilon_W,\\
(\varepsilon_X,\varepsilon_Y,  \varepsilon_Z, \varepsilon_W,  U)^\T\thicksim N(\boldsymbol{0},\boldsymbol{I}_5),\ \alpha_1,\beta_1,\gamma_1,\eta_1 \text{ are nonzero}.
\end{gathered}
\end{equation}
Under this model, we have 
\begin{equation}\nonumber
\sigma_{XY}\sigma_{ZW}-\sigma_{XW}\sigma_{YZ} = \alpha_2\beta_2\gamma_1\eta_1\operatorname{var}(U)\operatorname{var}(U^2),
\end{equation}
which is nonzero if nonlinear effects of $U$ exist, that is, 
$\alpha_2  $ and $\beta_2$ are nonzero.
As a result,  the classical tetrad constraint no longer holds 
even if $\H_0$ is correct.
\end{example} 
This example reveals potentially non-negligible type I error for testing 
$\H_0$ based on the classical tetrad constraint in nonlinear models.
Moreover, as we illustrate later in Example~\ref{ex:3}, the classical tretrad constraint can also lead to low testing power for distinguishing  from $\H_0$ in  nonlinear models.
These issues  may cause incorrect identification and interpretation of the model structure in practice.
In the next section, we introduce a generalized tetrad constraint, 
which aims to  establish testable implications for $\H_0$ under nonlinear models.

\section{Introducing the generalized tetrad constraint}\label{sec:intro}

In this section, we introduce a generalized tetrad constraint that extends the classical tetrad constraint to nonlinear and nonparametric models.
We first present the theoretical foundation of the generalized constraint, establishing testable implications for conditional independence given a latent variable under weaker assumptions.
We then propose a one-dimensional measure to quantify the deviation from this constraint, which facilitates the construction of a practical statistical test.

\subsection{The generalized tetrad constraint}

In order to test $\H_0$, it is crucial to first characterize the confounding effects  of $U$ on the observed variables $(X,Y,Z,W)$ and  the relationship between  the confounding effects.
However, in nonlinear models, the  confounding effects  are complicated, which cannot be characterized by regression coefficients as in linear models.
We consider the model characterized by the following assumptions.
\begin{assumption}[Confounding bridge]\label{assum:1}
(i) There exists some square-integrable function $h_0=h_0(w)$ such that for all $z$,
\begin{equation}\label{eq:3}
E(Y\mid U, Z=z)= E\{h_0(W)\mid U, Z=z\}. 
\end{equation}
(ii) There exists some square-integrable function $g_0=g_0(z)$ such that for all $w$,
\begin{equation}\label{eq:4}
E(Y\mid U,W=w)= E\{g_0(Z)\mid U, W=w\}.
\end{equation}
\end{assumption} 

\begin{assumption}[Completeness]\label{assum:2}
(i) For any square-integrable function $l$, $E\{l(W) \mid Z\} = 0$ almost surely if
and only if $l(W) = 0$ almost surely; 
(ii) for any square-integrable function $l$, $E\{l(Z) \mid W\} = 0$ almost surely if
and only if $l(Z) = 0$ almost surely.
\end{assumption}

Assumption~\ref{assum:1}(i) states that 
the confounding effects of $U$ on $Y$ and on a transformation
$h_0(W)$ of $W$
are the same within each level of $Z$. 
Therefore,  the function  $h_0$   captures the relationship between the confounding effects on  $Y$ and on  $W$, 
which we refer to as the confounding bridge function.
The interpretation of Assumption~\ref{assum:1}(ii) is analogous to that of Assumption~\ref{assum:1}(i).
The confounding bridge function is  introduced by \citet{miao2018identifying},
which has been used in causal inference for handling  unobserved nonlinear confounding effects \citep[for example,][]{shi2020multiply,cui2024semiparametric}.
Assumption~\ref{assum:2} is known as the completeness condition. Under Assumption~\ref{assum:2}, there is at most one solution to each of \eqref{eq:3} and \eqref{eq:4}. 
The completeness condition combined with Assumption~\ref{assum:1} states that the confounding bridge functions $h_0,g_0$ are unique. The completeness condition can accommodate both categorical and continuous distributions and intuitively implies that, any variability in $W$ is captured by variability in $Z$, and vice versa. 
In the categorical case, it requires that $Z$ and $W$ have the same number of categories.
In the continuous case, many commonly-used parametric and semiparametric models imply the completeness condition, for example, Assumption~\ref{assum:2} is satisfied if the distribution of $(Z,W)$ is from exponential families \citep{newey2003instrumental}.
Furthermore, \citet{d2011completeness} and \citet{darolles2011nonparametric} introduce useful results to justify the completeness condition for nonparametric models.
Completeness conditions play an important role in identification problems and have been widely invoked in instrumental variable and negative control methods for causal inference \citep{newey2003instrumental,miao2018identifying,cui2024semiparametric}.

Assumptions~\ref{assum:1} and~\ref{assum:2} together characterize the model for the joint distribution we consider,
which admits nonlinear structural models where  the conditional mean of the observed variables given  the unmeasured confounder could be nonlinear and unrestricted.
Assumption~\ref{assum:1}(i)  implies that 
\begin{equation}\label{eq:5}
E(Y\mid  Z)= E\{h_0(W)\mid Z\},
\end{equation}
and thus, the confounding bridge function also captures the relationship between the
observed associations between $(Y,Z)$ and between $(W,Z)$.
Equation \eqref{eq:5} is a Fredholm integral equation of the first kind. 
Assumption~\ref{assum:1}(i) assures existence of the solution, and
the completeness  of $f(W \mid  Z)$ in Assumption~\ref{assum:2}(i) guarantees uniqueness of the solution.
Equation \eqref{eq:5} offers a feasible strategy to identify the confounding bridge function with observed variables.
Based on the confounding bridge functions, we obtain the following result.
\begin{theorem}\label{thm:1}
Under Assumptions~\ref{assum:1} and~\ref{assum:2}, 
there exist unique square-integrable functions $h_0$ and $g_0$ satisfying the following equations almost surely:
\[
E(Y \mid Z) = E\{h_0(W) \mid Z\}, \quad E(Y \mid W) = E\{g_0(Z) \mid W\}.
\]
Moreover, these functions are equal to  the confounding bridge functions defined in Equations \eqref{eq:3} and \eqref{eq:4}, which are therefore identifiable from the observed data.
If in addition $\H_0$ is correct, then $h_0$ and $g_0$ also satisfy the following equation almost surely,
\begin{equation}\label{eq:6}
E(Y \mid X) = E\{h_0(W) \mid X\} = E\{g_0(Z) \mid X\}.
\end{equation}
\end{theorem}
Theorem~\ref{thm:1} demonstrates the identifiability of the confounding bridge functions and characterizes the relationship between the four observed variables under $\H_0$ with the confounding bridge functions.
It states that under $\H_0$, $Y$ and the confounding bridge functions  $g_0(Z)$ and $h_0(W)$ have equal conditional expectations given $X$. 
Together, this theorem establishes a necessary and testable condition \eqref{eq:6} for $\H_0$,
which we call  the generalized tetrad constraint.

\begin{example}[Continuum of Example~\ref{ex:1}]\label{ex:2}
Recall the setting in  Example~\ref{ex:1},
where the structural equation model \eqref{eq:2} is nonlinear and $\H_0$ is correct.
Under this model, Assumptions~\ref{assum:1} and~\ref{assum:2} hold, 
and we have confounding bridge functions $g_0(z)=(\beta_1\gamma_1^{-1}-2\beta_2\gamma_0\gamma_1^{-2})z+\beta_2\gamma_1^{-2}z^2+c_1$, $h_0(w)=(\beta_1\eta_1^{-1}-2\beta_2\eta_0\eta_1^{-2})w+\beta_2\eta_1^{-2}w^2+c_2,$ where $c_1$ and $c_2$ are constants.
Then we can verify that the generalized tetrad constraint  \eqref{eq:6} is entailed by model \eqref{eq:2}, 
that is, it holds for all possible values of the coefficients in model \eqref{eq:2}.
Although, the classical  tetrad constraint does not hold as shown in Example~\ref{ex:1}.
\end{example}
Equation \eqref{eq:6} consists of conditional moment equations, 
which implies the following marginal moment condition,
\begin{equation}\label{eq:7}
\operatorname{cov}(Y,X)=\operatorname{cov}\{h_0(W),X\}=\operatorname{cov}\{g_0(Z),X\}.
\end{equation} 
Particularly,    if the confounding bridge functions are linear, 
we can equivalently express them with the covariances of observed variables and obtain  the following result.
\begin{proposition}\label{prop:2}
Equation \eqref{eq:7} reduces to the classical tetrad constraint 
when the confounding bridge functions $g_0,h_0$ are linear in $Z,W$, respectively.
\end{proposition} 
In this view, the   proposed tetrad constraint is a generalization of the classical one, 
but it also accommodates nonlinear models when the classical one fails to hold.
The following proposition presents familiar examples  when the confounding bridge functions are linear.
\begin{proposition}\label{prop:3}
Under Assumptions~\ref{assum:1} and~\ref{assum:2}, if in addition $(X,Y,Z,W)$ follows  a multivariate elliptical  distribution, 
which has characteristic function $\phi_{(X,Y,Z,W)}(\boldsymbol{t})=e^{i\boldsymbol{t}^{\T}\boldsymbol{\mu}}\psi(\boldsymbol{t}^{\T}\boldsymbol{\Sigma} \boldsymbol{t})$ for any $\boldsymbol{t}$ with a location parameter $\boldsymbol{\mu}$, a positive-definite scale matrix $\boldsymbol{\Sigma}$ and  a univariate function $\psi$, 
then  the confounding bridge functions  $g_0(z),h_0(w)$ are linear in $z,w$, respectively.
Furthermore, Equation \eqref{eq:6} and the classical tetrad constraint are equivalent in this case.
\end{proposition} 
The family of multivariate elliptical distributions includes multivariate normal, multivariate $t$-distribution, symmetric multivariate Laplace distribution, and many other  distributions.
From Proposition~\ref{prop:3}, under such  elliptical distributions,
the generalized tetrad constraint is equivalent   to the classical one as we show in Section~\ref{ssec:proof-prop:3} of the Supplementary Material.

When the confounding bridge functions are linear, 
if $\H_0$ is correct, 
the generalized tetrad constraint implies   the classical one.
However, in general,
the relationship between two constraints is not easy to study,  
because the generalized tetrad constraint involves solving the integral equations, which are complicated  inverse problems,
and the confounding bridge functions may not have  closed-form expressions.
When  $\H_0$ is incorrect and the model is nonlinear,
there is no guarantee that  the generalized tetrad constraint fails to hold as long as  the classical one fails to hold.
Although $\H_0$ is symmetric about $(X,Y,Z,W)$,
the form of generalized tetrad constraint is asymmetric and a different  labeling   formulates a different constraint.
The form of the classical tetrad constraint is symmetric, and these two constraints are not equivalent  in general.
If for every permutation   of the four observed variables, the  corresponding confounding bridge functions exist and are unique, then a stronger necessary condition of $\H_0$ is that Equation \eqref{eq:6} holds for every permutation.
However, the requirement that Assumptions~\ref{assum:1} and~\ref{assum:2} hold for every permutation may be overly restrictive, 
and we formally define the generalized tetrad constraint for only one permutation of the four variables. 
In practice, one may consider using the generalized tetrad constraint for different permutations to strengthen the inference about whether $\H_0$ is correct.

In nonlinear models, neither the classical nor the generalized tetrad constraint is sufficient for  $\H_0$, 
that is, $\H_0$ failing to hold does not necessarily lead to violation of either the classical tetrad constraint or the generalized one.
Nonetheless, in certain situations such as the following example, 
the generalized tetrad constraint can be more informative about the departure from  $\H_0$.

\begin{example}\label{ex:3}
Consider the following model,   \begin{equation}\label{eq:8}
  \begin{gathered}
      X=\alpha_0+\alpha_1U+\varepsilon_X,\ Y=\beta_0+\beta_1U+\beta_{12} UX+\varepsilon_Y,\\
    Z=\gamma_0+\gamma_1U+\varepsilon_Z,\ W=\eta_0+\eta_1U+\varepsilon_W,\\
    (\varepsilon_X,\varepsilon_Y,  \varepsilon_Z, \varepsilon_W,  U)^\T\thicksim N(\boldsymbol{0},\boldsymbol{I}_5),\ \beta_{12}\neq0. 
  \end{gathered} 
  \end{equation}
Then we have $\sigma_{XW}\sigma_{YZ}=\sigma_{XY}\sigma_{ZW}=\sigma_{XZ}\sigma_{YW}$,
that is,  the classical tetrad constraint holds for all possible values of the parameters. 
Thus,   the classical tetrad constraint could not detect the deviation  from $\H_0$ in this nonlinear model.
However, we obtain the  confounding bridge functions $g_0(z)=(\beta_1+\alpha_0\beta_{12}-2\alpha_1\beta_{12}\gamma_0\gamma_1^{-1})\gamma_1^{-1}z+\alpha_1\beta_{12}\gamma_1^{-2} z^2+d_1 $ and $h_0(w)=(\beta_1+\alpha_0\beta_{12}-2\alpha_1\beta_{12}\eta_0\eta_1^{-1})\eta_1^{-1}w+\alpha_1\beta_{12}\eta_1^{-2} w^2+d_2$, where $d_1,d_2$ are constants.
We can verify that $E\{Y-h_0(W)\mid X\}\neq 0$,
that is, the generalized tetrad constraint does not  hold  generically except for  a set of parameter values having zero Lebesgue measure.
\end{example}

As a conclusion of the above discussion, the generalized tetrad constraint \eqref{eq:6}  is always a necessary condition for $\H_0$ under Assumptions~\ref{assum:1} and~\ref{assum:2}, 
which are much less restrictive than the linearity assumption underpinning the classical tetrad constraint;
however, the classical tetrad constraint is no longer a necessary condition for $\H_0$ in nonlinear models.
In general \eqref{eq:6} is not a sufficient condition for $\H_0$. 
When the confounding bridge functions are linear,
the generalized tetrad constraint is at least as informative as the classical one.
Moreover, it shows the power to distinguish the deviation of a nonlinear model from $\H_0$ in certain situations while the classical one fails.
These  results allow us to test $\H_0$ based on the generalized tetrad constraint,
which is a set of conditional  moment equations.
A straightforward approach for testing  such conditional moment equations  is to first estimate the corresponding conditional distributions and the confounding bridge functions, 
and then comparing the conditional expectations in \eqref{eq:6}.
However, this approach is complicated and may lack power
due to slow convergence for estimating conditional distributions.
In the  next subsection, we consider an alternative approach   based on 
a measure of generalized tetrad difference,   
which summarizes the difference between two conditional expectations   with a one-dimensional measure, 
without information loss.

\subsection{A measure for the generalized tetrad difference}
\label{sec:MGT}
Consider the generalized tetrad difference $E\{Y-h_0(W)\mid X\}$.
In order to measure this difference in conditional expectations,
note that   \eqref{eq:6} holds if and only if the following marginal moment equations hold,
\begin{equation}\label{eq:9}
\operatorname{cov}\{Y-h_0(W),e^{isX}\}=\operatorname{cov}\{Y-g_0(Z),e^{isX}\}=0\ \text{for all $s\in\mathbb{R}$}.
\end{equation}
In principle, one can  measure the generalized tetrad difference
by assessing how far away the covariances in \eqref{eq:9} depart from zero, 
which is   however not possible for all $s\in \mathbb{R}$. 
An intuitive approach is to average the covariances over all   $s$,
which leads to the  following measure of generalized tetrad (MGT) difference,
\begin{equation}\label{eq:10}
\operatorname{MGT}(h_0)=\Biggl(\pi^{-1}\int_{s\in\mathbb{R}}\frac{\bigl\vert E\bigl[\{Y-h_0(W)\}e^{isX}\bigr]-E\{Y-h_0(W)\}E(e^{isX})\bigr\vert^2}{s^2}\dif s\Biggr)^{1/2}.
\end{equation}
Analogously, we can define MGT($h$) for any square-integrable function $h$.
The MGT  is a one-dimensional nonnegative functional of $h$. 
For the confounding bridge function $h_0$ that satisfies   \eqref{eq:3},
we have $E\{Y-h_0(W)\}=0$ and the $\operatorname{MGT}(h_0)$ in \eqref{eq:10} reduces to
\[\operatorname{MGT}(h_0)=\biggl(\pi^{-1}\int_{s\in\mathbb{R}}\bigl\vert E\bigl[\{Y-h_0(W)\}e^{isX}\bigr]\bigr\vert^2s^{-2}\dif s\biggr)^{1/2}.\]

The  MGT is motivated   by the martingale difference divergence \citep[MDD,][]{shao2014martingale}  for characterizing the conditional mean dependence between two random variables. 
It is the nonnegative square root of a weighted integral of squared correlations between the unconditional variable and the Fourier bases of the conditional variable.
Here in \eqref{eq:10}, the unconditional variable is $Y-h_0(W)$
and the conditional variable is $X$.
However, the MGT is different from the MDD that $h_0$ is unknown and has to be estimated first.
The integral in \eqref{eq:10} has a closed form, which is much more convenient to calculate, see Proposition~\ref{prop:4} below.

\begin{proposition}\label{prop:4} 
Suppose   $E\{\vert Y-h_0(W)\vert^2+\vert X\vert^2\}<\infty$, then 

(a) $\operatorname{MGT}(h_0)^2=-E\bigl[\{Y-h_0(W)\}\{Y^\prime-h_0(W^\prime)\}\vert X-X^\prime\vert\bigr],$
where $(X^\prime,Y^\prime,W^\prime)$ is an i.i.d. copy of $(X,Y,W)$;

(b) $\operatorname{MGT}(h_0)=0$ if and only if $E\{Y-h_0(W)\mid X\}=0$ almost surely.
\end{proposition}
Proposition~\ref{prop:4}  can be   obtained similarly as Theorem 1 in \citet{shao2014martingale}. 
Proposition~\ref{prop:4}(a) states that the MGT can be expressed in  marginal expectation, which will be used as the basis for the later calculations.
Proposition~\ref{prop:4}(b) shows that  a zero $\operatorname{MGT}(h_0)$ is equivalent to a zero generalized tetrad difference.
Analogously, we can  define MGT for $g_0$ and aggregate them together to obtain a measure for the deviation from the generalized tetrad constraint, called the aggregated measure of generalized tetrad differences (AMGT),
\begin{eqnarray*}
\operatorname{AMGT}(h_0,g_0)^2&=&\operatorname{MGT}(h_0)^2+\operatorname{MGT}(g_0)^2\\
&=&\pi^{-1}\int_{s\in\mathbb{R}}\bigr\Vert E\bigr[\{Y-h_0(W),Y-g_0(Z)\}^{\mathrm T} e^{isX}\bigr]\bigr\Vert^2s^{-2}\dif s.
\end{eqnarray*}
By Proposition~\ref{prop:4}(b) and an analogous result for $\operatorname{MGT}(g_0)$,
we have $\operatorname{AMGT}(h_0,g_0)=0$ if and only if the generalized tetrad constraint \eqref{eq:6} holds almost surely.
Therefore,  we can  assess the  generalized tetrad differences using the AMGT, 
which does not lose any information but enjoys the simplicity for calculation and inference.
Based on  Proposition~\ref{prop:4},
we next consider the estimation of the AMGT and then use it to test the null hypothesis $\H_0$ with data samples.

\section{The generalized tetrad test}\label{sec:tetrad-test}

In this section, we develop a formal statistical testing procedure based on the generalized tetrad constraint and the AMGT introduced in Section~\ref{sec:intro}.
We first construct a testing framework using the AMGT and derive the asymptotic distribution of the test statistic. 
We then discuss the estimation of the required nuisance models, namely the confounding bridge functions, under both parametric and nonparametric models.

\subsection{A statistical testing framework based on the AMGT}
After establishing the necessary condition for $\H_0$,
we aim to develop a formal statistical test for the null hypothesis by testing
whether the generalized tetrad constraint holds.
Suppose we have  $n$ i.i.d. observations of $(X,Y,Z,W)$, denoted by $\boldsymbol{O}_j=(X_j,Y_j, Z_j, W_j)$, $j=1,\ldots,n$.
According to Proposition~\ref{prop:4}(b), we can  test $\H_0$   by testing  $\operatorname{AMGT}(h_0,g_0)=0$.
To construct an estimator of the AMGT, we first obtain an estimator of the confounding bridge functions, and then plug them into the sample version of the closed-form expression of the AMGT. 
For simplicity, we denote $R(h)=Y-h(W)$, 
$R_j(h)=Y_j-h(W_j) $ and $\Lambda_s(h)=E\{R(h)e^{isX}\}-E\{R(h)\}E(e^{isX})$,
where $R(h)$ captures the  difference between $Y$ and the confounding bridge function $h(W)$ and  is called the residual function with unknown parameter $h$.

We first obtain the confounding bridge  estimator $\widehat{h}_0$, 
and  the estimator of the $\operatorname{MGT}(h_0)$  is 
\begin{align}\label{eq:11}
\operatorname{MGT}_n(\widehat{h}_0)&=\left\{\pi^{-1}\int_{s\in\mathbb{R}} \bigl\vert \widehat{\Lambda}_s(\widehat{h}_0)\bigr\vert^2s^{-2} \dif s\right\}^{1/2} \nonumber \\
&= \left[-\frac{1}{n^2}\sum_{j,k=1}^n\left\{R_j(\widehat{h}_0)-n^{-1}\sum_{l=1}^nR_l(\widehat{h}_0)\right\}\left\{R_k(\widehat{h}_0)-n^{-1}\sum_{l=1}^nR_l(\widehat{h}_0)\right\}\Delta_{jk}\right]^{1/2},
\end{align}
where $\widehat{\Lambda}_s(\widehat{h}_0)=n^{-1}\sum_{j=1}^nR_j(\widehat{h}_0)e^{isX_j}-n^{-1}\sum_{j=1}^nR_j(\widehat{h}_0)n^{-1}\sum_{j=1}^ne^{isX_j}$, and $\Delta_{jk}=\vert X_j-X_k\vert-n^{-1}\sum_{j=1}^n\vert X_j-X_k\vert-n^{-1}\sum_{k=1}^n\vert X_j-X_k\vert+n^{-2}\sum_{j,k=1}^n\vert X_j-X_k\vert.$
One can analogously obtain an estimator $\operatorname{MGT}_n(\widehat{g}_0)$ of $\operatorname{MGT}(g_0)$, and an estimator of $\operatorname{AMGT}(h_0,g_0)^2$ is the summation of $\operatorname{MGT}_n(\widehat{h}_0)^2$ and $\operatorname{MGT}_n(\widehat{g}_0)^2$.
In order to establish the asymptotic distribution of $\operatorname{AMGT}_n(\widehat{h}_0,\widehat{g}_0)^2$, 
we require the following high-level assumptions about the observed-data distribution and the nuisance estimators.
\begin{assumption}\label{assum:3}
All observed variables $X,Y,Z$ and $W$ have compact supports.
\end{assumption}
Assumption~\ref{assum:3} is a standard condition on the distribution of the observed variables.
The compactness is not restrictive and often holds in practice where  many observed variables (for example, age, socioeconomic status, scores, responses to a survey) are bounded.
The assumption can be relaxed  to situations where  the distributions of the observed variables have unbounded supports with sufficiently thin tails, such as the multivariate normal distribution, as shown in the simulation studies.

The following two assumptions about the confounding bridge function and its estimator are required to be satisfied for both $h_0$ and $g_0$. 
\begin{assumption}\label{assum:4}
For $h=h_0$ and $\widehat{h}_0$, $E\{\vert h(W)\vert^2\}<\infty$ and $E\{\vert R(h)\vert^4 \}<\infty$; $h_0$ and $\widehat{h}_0$ are in a Donsker class.
We also assume the parallel conditions for $g_0$ and $\widehat{g}_0$, with $W$ replaced by $Z$ everywhere.
\end{assumption}

\begin{assumption}\label{assum:5}
$\widehat{h}_0$ is an estimator of $h_0$ such that 

\begin{enumerate}
\item[(i)] $\Vert\widehat{h}_0-h_0\Vert_{\infty}=o_p(1)$, and $E(\widehat{h}_0-h_0)^2=o_p(n^{-1/2})$;

\item[(ii)]  $n^{1/2}\{\Lambda_s(\widehat{h}_0)-\Lambda_s(h_0)\}=n^{-1/2}\sum_{j=1}^n\varphi_s(\boldsymbol{O}_j;h_0)+r_{1,n}(s)$, where  $\sup_{s}\vert r_{1,n}(s)\vert=o_p(1) $,   $\varphi_s$ is a measurable function that may depend on $s$ and satisfies $$E\{\varphi_s(\boldsymbol{O};h_0)\}=0 , ~~ E\{\vert\varphi_s(\boldsymbol{O};h_0)\vert^2\}\leq M<\infty \text{ for all } s.$$  
The covariance function $c_2(s_1,s_2)=E\{\varphi_{s_1}(\boldsymbol{O};h_0)\overline{\varphi}_{s_2}(\boldsymbol{O};h_0)\}$ is uniformly continuous in $s_1,s_2$, satisfying   
$0<\int_{s\in \mathbb{R}}c_2(s,s)s^{-2}\dif s<\infty.$
\end{enumerate}
We also assume the parallel conditions for $g_0$ and $\widehat{g}_0$, with $W$ replaced by $Z$ everywhere.
\end{assumption}
Assumption~\ref{assum:4}  requires the confounding bridge function and its estimator to be square-integrable, and both the true and the estimated residual functions have finite fourth moments. 
It also imposes the Donsker condition \citep{van1996weak} on the complexity of the nuisance models, 
which could be relaxed to admit more complicated working models by employing the cross-fitting method \citep[for example,][]{chernozhukov2018double}.
Assumption~\ref{assum:5}(i) requires the consistency of the estimator $\widehat{h}_0$ in the supremum norm, and its convergence rate in $L_2$ norm is faster than $n^{1/4}$.
Assumption~\ref{assum:5}(ii) indicates that $\Lambda_s(h_0)$ can be well approximated by $\Lambda_s(\widehat{h}_0)$ uniformly in $s$.
It can be verified for different estimating methods under parametric or nonparametric models of the confounding bridge function.
In practice, one can specify a parametric model for the confounding bridge function and apply standard  estimation approach such as the generalized method of moments \citep[GMM,][]{hansen1982large}.
In this case, Assumption~\ref{assum:5}(ii) is satisfied under certain   regularity conditions;
$\varphi_{s}(\boldsymbol{O};h_0)$ is the influence function of $\Lambda_s(\widehat{h}_0)$
and $c_2(s_1,s_2)$ is the covariance of the influence functions $\varphi_{s_1}(\boldsymbol{O};h_0)$ and $\varphi_{s_2}(\boldsymbol{O};h_0)$.
Alternatively, one can apply a nonparametric method,  such as  the penalized sieve minimum distance (PSMD) method, to estimate $h_0$. 
It can be used  for   estimation of unknown functions satisfying a general conditional moment restriction \citep{chen2012estimation}. 
\citet{chen2015sieve}  establish  the asymptotic normality for plug-in PSMD estimators of functionals of the unknown function.
However, their method requires that the pathwise derivative of the target functional at $h_0$ is nonzero, and cannot be directly applied to the MGT estimation.
Following the spirit of \citet{chen2015sieve},
Assumption~\ref{assum:5}(ii) can be verified for functionals $\Lambda_s$ instead.

The proposed MGT and AMGT estimators are consistent.
However, as shown in Lemma~\ref{lem:mgt} in Section~\ref{ssec:aym-mgt} of the Supplementary Material, under $\H_0$, the asymptotic distribution of $n\operatorname{MGT}_n(\widehat{h}_0)^2$ is a weighted integral of the norms of a complex Gaussian process $\Gamma(s)$, which is nonstandard and complicated for making inference, and therefore,  we consider a standardized version of the AMGT estimator.
Let 
\[T_{n}(\widehat{h}_0,\widehat{g}_0)=n\operatorname{AMGT}_n(\widehat{h}_0,\widehat{g}_0)^2/S_n,\]
where 
\begin{equation}\label{eq:12}
\begin{aligned}
S_n&=S_n(\widehat{h}_0)+S_n(\widehat{g}_0),\\
S_n(\widehat{h}_0)&=n^{-1}\pi^{-1}\int_{s\in \mathbb{R}}\sum_{j=1}^n\Bigl\vert R_j(\widehat{h}_0)e^{isX_j}-R_j(\widehat{h}_0)n^{-1}\sum_{k=1}^ne^{isX_k}+\widehat{\varphi}_s(O_j;\widehat{h}_0)\Bigr\vert^2s^{-2}\dif s.
\end{aligned}
\end{equation}
We standardize the AMGT estimator so that the asymptotic distribution of $T_{n}(\widehat{h}_0,\widehat{g}_0)$ under $\H_0$ has a more interpretable expression with expectation equal to one.
For establishing  asymptotic distribution of $T_{n}(\widehat{h}_0,\widehat{g}_0)$ under $\H_0$, we assume the following condition, 
which only requires $S_n(\widehat{h}_0)$ to be consistent under $\H_0.$
Note that a parallel condition for $S_n(\widehat{g}_0)$ is also required.
Besides, we can also replace the requirements in Assumption~\ref{assum:5}(ii) with weaker conditions that only need to hold under $\H_0.$
\begin{assumption}\label{assum:6}
The estimator $S_n(\widehat{h}_0)=o_p(n)$, and converges to $\pi^{-1}\int_{s\in\mathbb{R}}\operatorname{cov}_{\Gamma}(s,s)s^{-2}\dif s$  in probability under $\H_0$, where $\operatorname{cov}_{\Gamma}$ is the covariance function of the Gaussian process $\Gamma$ with the specific form detailed in Section~\ref{ssec:aym-mgt} of the Supplementary Material.
We also assume the parallel conditions for $S_n(\widehat{g}_0)$.
\end{assumption}
\begin{theorem}\label{thm:5}
Under Assumptions~\ref{assum:1},~\ref{assum:2},~\ref{assum:4} and~\ref{assum:5}(i), we have

(a) $\operatorname{AMGT}_n(\widehat{h}_0,\widehat{g}_0)\to\operatorname{AMGT}(h_0,g_0)$ in probability;

(b) if further Assumptions~\ref{assum:3},~\ref{assum:5}(ii) and~\ref{assum:6}  hold, and $\H_0$ is correct, then
we have
$T_n(\widehat{h}_0,\widehat{g}_0)$ converges in distribution to a random variable $Q=\sum_{j=1}^\infty\lambda_j\chi_{1j}^2$, 
where $\chi_{1j}^2$ are i.i.d.  Chi-square random variables with one degree of freedom,
and $\{\lambda_j\}_{j=1}^\infty$ are nonnegative constants that depend on the distribution of $(X,Y,Z,W)$ and $\sum_{j=1}^\infty\lambda_j=1$. 
\end{theorem}
Theorem~\ref{thm:5}(a) shows consistency of the AMGT estimator,
and Theorem~\ref{thm:5}(b)    shows that after standardized by $S_n$,  
the AMGT estimator converges in distribution to a weighted sum of a series of independent Chi-square variables $\chi_{1j}^2$ with one degree of freedom under $\mathbb{H}_0$.
However, the exact distribution of $Q$ is complicated, and it remains a challenge to directly use the result for inference.
Instead, we consider an approximation of $Q$  using a   Chi-square variable with one degree of freedom, 
which is motivated by a well-known property of the weighted sum of Chi-square distributions:   $\operatorname{pr}(Q\geq z^2_{1-\alpha/2})\leq\alpha$ for any $0<\alpha\leq0.215$  \citep[Theorem 1]{szekely2003extremal}, where $z_{1-\alpha/2}$ is the $(1-\alpha/2)$th quantile of the standard normal distribution.
Based on this result, we can use $T_n(\widehat{h}_0,\widehat{g}_0)$ as the test statistic:
given the  significance level $\alpha$, we reject $\mathbb{H}_0$ if and only if $T_n(\widehat{h}_0,\widehat{g}_0)\geq z^2_{1-\alpha/2}$.

\begin{proposition}\label{prop:6}
Under Assumptions~\ref{assum:1}--\ref{assum:6},  we have

(a) if $\H_0$ is correct, 
then $\lim_{n\to\infty}\operatorname{pr}\bigl\{T_n(\widehat{h}_0,\widehat{g}_0)\geq z^2_{1-\alpha/2}\bigr\}\leq\alpha$ for all $0 <\alpha \leq 0.215$;

(b) if the generalized tetrad constraint \eqref{eq:6} does not hold, that is, $E\{Y-h_0(W)\mid X\}\neq0$ or $E\{Y-g_0(Z)\mid X\}\neq0$,
then $\lim_{n\to\infty}\operatorname{pr}\bigl \{T_n(\widehat{h}_0,\widehat{g}_0)\geq z^2_{1-\alpha/2}\bigr\}=1$.
\end{proposition}

Proposition~\ref{prop:6}(a) shows that  for any   significance level below   $0.215$, the type I error of the proposed test  can be   controlled asymptotically.
Proposition~\ref{prop:6}(b) shows that  the proposed test has power approaching unity for distinguishing violation of the generalized tetrad constraint.
However,  the proposed test may not reject $\H_0$ even if $\H_0$ is incorrect, because the generalized tetrad constraint is only a necessary but not sufficient condition for $\H_0$.
In addition, the generalized tetrad test is conservative due to the Chi-square approximation to the asymptotic distribution of the test statistic, as shown in the simulation study. 
An alternative approach is to use a bootstrap procedure to approximate the distribution of the test statistic under $\H_0$, which may improve finite-sample inference.
This idea has been used by, for example,  \citet{zhang2020lack}; we refer to this  paper for a detailed discussion of the bootstrap's theoretical properties and general applicability. 
Nevertheless, the extent of efficiency gain from bootstrapping and its theoretical guarantees in this specific setting require further investigation.

In practice, there may exist  observed covariates $\boldsymbol{V}$ that   also account for the  associations between  $(X,Y,Z,W)$, and we are interested in testing  whether the observed associations can be explained away by the latent variable $U$ after
adjusting for observed covariates, that is,   $\H^\prime_0: X\ind Y\ind Z\ind W\mid (U,\boldsymbol{V})$.
In this case, we replace $X$, $Z$, $W$ with vectors $(X,\boldsymbol{V})$, $(Z,\boldsymbol{V})$ and $(W,\boldsymbol{V})$ of dimension $d>1$, respectively, in Assumptions~\ref{assum:1},~\ref{assum:2} and Equation \eqref{eq:6}.
In addition, we replace the variable $s$ with a $d$-dimensional vector $\boldsymbol{s}$
in the AMGT, and  $\pi$, $s^{-2}$ are replaced by $c_d=\pi^{(d+1)/2}/\Gamma((d+1)/2)$ and $\Vert\boldsymbol{s}\Vert^{-(d+1)}$, respectively.
With these changes, the results in Theorem~\ref{thm:1} and Proposition~\ref{prop:4} are still valid.
We can analogously estimate the confounding bridge functions and test $\H^\prime_0$ based on the AMGT, with appropriate adjustment for multi-dimensional variables.

\subsection{Estimation of   nuisance models}\label{ssec:estimationnuisanceparameters}
We now briefly illustrate how to use specific estimating methods to estimate $h_0$, 
including the  generalized method of moments (GMM) and the penalized sieve minimum distance (PSMD) method,
when   parametric or nonparametric working models for $h_0$ are used, respectively.
Assumptions~\ref{assum:3}--\ref{assum:6}   about the observed-data distribution and the nuisance estimators can be met  for these estimation methods  and models.

For parametric estimation, we first specify
a parametric model for the confounding bridge function $h=h(W;\boldsymbol{\theta})$ with a finite-dimensional parameter $\boldsymbol{\theta}\in\boldsymbol{\Theta}\subset\mathbb{R}^p$, and the true value $\boldsymbol{\theta}_0$ satisfies $h_0=h(W;\boldsymbol{\theta}_0)$.
Then an estimator of $\boldsymbol{\theta}_0$ can be obtained by solving the following estimating equations,
\begin{equation}\nonumber
m_n(\boldsymbol{\theta})=\frac{1}{n}\sum_{j=1}^n\boldsymbol{J}(\boldsymbol{O}_j;\boldsymbol{\theta})=\boldsymbol{0},\ \boldsymbol{J}(\boldsymbol{O}_j;\boldsymbol{\theta})=\left\{Y_j-h(W_j;\boldsymbol{\theta})\right\}\boldsymbol{B}(Z_j),
\end{equation}
where $\boldsymbol{B}(z)$ is a user-specified vector function with dimension no smaller than that of $\boldsymbol{\theta}$. 
A standard approach to estimate $\boldsymbol{\theta}_0$ is the GMM \citep{hansen1982large}, which solves
\begin{equation}\label{eq:13}
\widehat{\boldsymbol{\theta}}=\arg\min_{\boldsymbol{\theta}}m_n^{\T}(\boldsymbol{\theta})\boldsymbol{\Omega} m_n(\boldsymbol{\theta}),\ \text{with $\boldsymbol{\Omega}$ a user-specified positive-definite weight matrix.} 
\end{equation}
Under certain regularity conditions \citep{hansen1982large,hall2005generalized}, the estimator $\widehat{\boldsymbol{\theta}}$ is consistent and asymptotically normal provided that the model for the confounding bridge function is correctly specified.
After obtaining $\widehat{\boldsymbol{\theta}}$,
we estimate $h_0$ with $\widehat{h}_0(\cdot)=h(\cdot;\widehat{\boldsymbol{\theta}})$, which is consistent and satisfies the requirements of Assumption~\ref{assum:5}, with 
\begin{align*}
\varphi_s(\boldsymbol{O};h_0)&=\bigl[E\{\nabla_{\boldsymbol{\theta}}h(W;\boldsymbol{\theta}_0)e^{isX}\}-E\{\nabla_{\boldsymbol{\theta}}h(W;\boldsymbol{\theta}_0)\}E(e^{isX})\bigr]^{\T}\boldsymbol{\Sigma}_1\boldsymbol{J}(\boldsymbol{O};\boldsymbol{\theta}_0),\\
\boldsymbol{\Sigma}_1&=(\boldsymbol{M}^{\T}\boldsymbol{\Omega} \boldsymbol{M})^{-1}\boldsymbol{M}^{\T}\boldsymbol{\Omega},\quad  \boldsymbol{M}=\lim_{n\to\infty}\frac{\partial m_n(\boldsymbol{\theta})}{\partial\boldsymbol{\theta}^{\T}}\biggr|_{\boldsymbol{\theta}=\boldsymbol{\theta}_0}.
\end{align*}
The estimator $S_n(\widehat{h}_0)$ can be constructed as in Equation~\eqref{eq:12} with
\begin{eqnarray*} \widehat{\varphi}_s(O_j;\widehat{h}_0)=\frac{1}{n}\biggl\{\sum_{k=1}^n\nabla_{\boldsymbol{\theta}}h(W_k;\widehat{\boldsymbol{\theta}})e^{isX_k}-\frac{1}{n}\sum_{k=1}^n\nabla_{\boldsymbol{\theta}}h(W_k;\widehat{\boldsymbol{\theta}})\sum_{l=1}^ne^{isX_l}\biggr\}^{\T}\widehat{\boldsymbol{\Sigma}}_1\boldsymbol{J}(\boldsymbol{O}_j;\widehat{\boldsymbol{\theta}}),
\end{eqnarray*}
where $\widehat{\boldsymbol{\Sigma}}_1$ is a consistent estimator of $\boldsymbol{\Sigma}_1$.
The estimator $S_n(\widehat{h}_0)$ satisfies the requirements of Assumption~\ref{assum:6} under certain regularity conditions.

For nonparametric estimation,  we restrict the confounding bridge function $h_0$ and its estimator in an infinite-dimensional function parameter space $\mathcal{H}$, 
which consists of smooth functions with certain requirements, such as a bounded Sobolev norm.
We  can apply the PSMD approach \citep{chen2015sieve} to estimate $h_0$,
by approximating $\mathcal{H}$ with linear expansions $\mathcal{H}_n=\left\{h\in \mathcal{H}:h(w)=\boldsymbol{\gamma}^\T\boldsymbol{\xi}^{q_n}(w)\right\}$ on growing finite-dimensional sieve bases (for example, wavelets, splines, Fourier series),
and  minimizing a penalized empirical minimum distance criterion over $\mathcal{H}_n$.
Under Assumptions~\ref{assum:1} and~\ref{assum:2}, the confounding bridge function $h_0$ is the unique minimizer of $Q(h)=E\left[E\left\{Y-h(W)\mid Z\right\}^2\right]$ in $\mathcal{H}$ , which is  the  criterion function.
The empirical criterion function $\widehat{Q}_n(h)= \sum_{j=1}^n\widehat{\rho}(Z_j,h)^2/n$ is constructed by estimating the conditional expectation $\rho(Z,h)=E\{Y-h(W)\mid Z\}$ using linear regression of $Y-h(W)$ on a series of $k_n\geq q_n$ approximating functions $\boldsymbol{p}^{k_n}(z)$ of $z$.
An approximate PSMD estimator  $\widehat{h}_0\in\mathcal{H}_n$ satisfies that
\begin{equation}\label{eq:14}
\widehat{Q}_n(\widehat{h}_0)+\lambda_n \operatorname{Pen}(\widehat{h}_0)\leq\inf\limits_{h\in\mathcal{H}_n}\{\widehat{Q}_n(h)+\lambda_n \operatorname{Pen}(h)\}+o_p(n^{-1}),
\end{equation}
where $\lambda_n \operatorname{Pen}(h) > 0$ is a penalty term such that $\lambda_n=o(1)$.
Under certain regularity conditions such as those described by \citet{chen2015sieve}, the estimator $\widehat{h}_0$ is consistent and satisfies the requirements of Assumption~\ref{assum:5}, with
$\varphi_s(\boldsymbol{O};h_0)=\{Y-h_0(W)\}E\{v_s^*(W)\mid Z\}$,
where $v_s^*\in\mathcal{H}$ is the limit of $v_{s,n}^*=\boldsymbol{\xi}^{q_n}(\cdot)^{\T}\boldsymbol{D}_n^- \bigl(\odif\Lambda_s(h_0)/\odif h[\boldsymbol{\xi}^{q_n}]\bigr)$ in the pseudo-metric  $\Vert h\Vert_{\operatorname{w}}=\bigl(E\bigl[\bigl\vert E\{h(W)\mid Z\}\bigr\vert^2\bigr]\bigr)^{1/2}$, $\boldsymbol{D}_n^-$ is the generalized inverse of the matrix $\boldsymbol{D}_n=E\bigl[E\{\boldsymbol{\xi}^{q_n}(W)\mid Z\}E\{\boldsymbol{\xi}^{q_n}(W)\mid Z\}^{\T}\bigr]$.
The estimator $S_{2,n}(\widehat{h}_0)$ can be constructed as in Equation~\eqref{eq:12} with
\begin{eqnarray*}
\widehat{\varphi}_s(O_j;\widehat{h}_0)=\bigl\{Y_j-\widehat{h}_0(W_j)\bigr\}\widehat{E}\bigl\{\widehat{v^*_{s,n}}(W)\mid Z_j\bigr\},
\end{eqnarray*}
where $\widehat{v^*_{s,n}}=\boldsymbol{\xi}^{q_n}(\cdot)^{\T}\widehat{\boldsymbol{D}}_n^-\bigl(\odif{\widehat{\Lambda}_s(\widehat{h}_0)}/\odif h[\boldsymbol{\xi}^{q_n}]\bigr)$, $\widehat{\boldsymbol{D}}_n^-$ is the generalized inverse of the matrix $\widehat{\boldsymbol{D}}_n=n^{-1}\sum_{j=1}^n\widehat{E}\{\boldsymbol{\xi}^{q_n}(W)\mid Z_j\}\widehat{E}\{\boldsymbol{\xi}^{q_n}(W)\mid Z_j\}^{\T}$.
The estimator $S_{2,n}(\widehat{h}_0)$ satisfies the requirements of Assumption~\ref{assum:6} under certain low-level conditions such as those described by \citet{chen2015sieve}.

In addition of the PSMD method we introduce, there exist other nonparametric estimation methods for the confounding bridge functions, such as Reproducing Kernel Hilbert Space methods \citep{ghassami2025causal,wu2025discovering} and neural network methods \citep{kallus2021causal},
which can consistently estimate the confounding bridge functions under certain conditions.

\section{Simulation studies}
\label{sec:simulation}
We evaluate performance of the generalized tetrad test  and compare it to the test based on the classical tetrad constraint via simulations.
Data  are generated according to
\begin{equation}\label{eq:15}
\begin{gathered}
X=\alpha_0+\alpha_1U+\alpha_2U^2+\varepsilon_1,\ Y=\beta_0+\beta_1U+\beta_2U^2+\delta X+\varepsilon_2,\\
Z=\gamma_0+\gamma_1U+\varepsilon_3,\ W=\eta_0+\eta_1 U+\varepsilon_4,\ (U,\varepsilon_1,\varepsilon_2,\varepsilon_3,\varepsilon_4)^{\T} \thicksim N(\boldsymbol{0},\boldsymbol{I}_5),\\
(\alpha_0,\beta_0,\gamma_0,\eta_0)=(0.5,-1,0.5,1),\quad
(\alpha_1,\beta_1,\gamma_1,\eta_1)=(0.5,0.5,1.5,1).
\end{gathered}
\end{equation}
Table~\ref{tab:1} presents the   settings of  parameters we consider in simulations.
In Setting (I), the null hypothesis $\H_0$ is correct and the variables follow a linear  model.
In Setting (II), the null hypothesis $\H_0$ is correct and the model is  nonlinear.
The values of $(\alpha_2,\beta_2)$ in (II.a) and (II.b) indicate the degree of nonlinearity. 
The theoretical properties of the tetrad constraints in this setting have been discussed in Examples~\ref{ex:1} and~\ref{ex:2}.
In Setting (III), $\H_0$ is incorrect and the model is  linear. 
The value of $\delta$  indicates the  degree of deviation from $\H_0$.

\begin{table}[tb]
\centering
\caption{Settings of parameters in simulations} \label{tab:1}
{
\setlength{\aboverulesep}{0pt}
\setlength{\belowrulesep}{0pt}
\begin{tabular}{lllll}
\toprule
\multirow{3}{*}{$\H_0$ Correct} &(I)& & $\alpha_2=\beta_2=\delta=0$ & \\
\cmidrule{2-5}
&\multirow{2}{*}{(II)}& (a)& \multirow{2}{*}{$\delta=0$} &  $(\alpha_2,\beta_2)=(0.1,0.2)$\\
& & (b)&  &  $(\alpha_2,\beta_2)=(0.3,0.4)$\\
\midrule
\multirow{2}{*}{$\H_0$ Incorrect} &\multirow{2}{*}{(III)}& (a)& \multirow{2}{*}{$\alpha_2=\beta_2=0$} &    $\delta=0.15$\\
&&  (b)& &    $\delta=0.3$\\
\bottomrule
\end{tabular}
}
\end{table}

We apply both the classical tetrad test (CT) and  the proposed generalized tetrad test   for testing the null hypothesis $\H_0: X\ind Y\ind Z\ind W\mid U$ at significance level 0.05.
The classical tetrad test is based on the  method  by \citet{bollen1990outlier}, 
which  estimates  the   tetrad differences  by plugging in the sample covariance and conducts tests based on the  asymptotic normality of this estimator.
The resulting test statistic is $T=n\hat t^\T\hat\Sigma_{\hat t}^{-1}\hat t$, where $\hat t$ is an estimator of the tetrad differences, $\hat\Sigma_{\hat t}$ is an estimator of the covariance matrix of $\hat t$, and the test statistic has an asymptotic Chi-square distribution with two degrees of freedom under $\H_0$ and in linear models.
In the generalized tetrad test, we consider two different methods for estimation of the confounding bridge function, 
including  the GMM  and the PSMD method.
The corresponding tests are denoted by GT$_1$ and GT$_2$, respectively.
For GT$_1$,
we set $h(w;\boldsymbol{\theta})=\theta_0+\theta_1w$ and $\boldsymbol{B}(z)=(1,z)^{\T}$ for the GMM,
and for GT$_2$ we set $(\boldsymbol{\xi}^{q_n},\boldsymbol{p}^{k_n},\operatorname{Pen}(h),\lambda_n)$ to $(\operatorname{Pol}(4),\operatorname{Pol}(7),\Vert h\Vert_{L^2}^2+\Vert\nabla h\Vert_{L^2}^2,4\times10^{-5})$ for the PSMD approach,
where $\operatorname{Pol}(r)$ denotes power series up to $(r - 1)$th degree,
$\Vert h\Vert_{L^2}^2=E\{\vert h(W)\vert^2\},$
and $\nabla h$ denotes the derivative of $h(w)$ with respect to $w$.
As suggested by \citet{chen2015sieve},
we use power series as the sieve basis and a  small fixed value as the penalty $\lambda_n$  due to the computational convenience,
and we refer to this paper for further guidance on the choice of sieve basis and the penalty.
The estimation of $g_0(z)$ in GT$_1$ and GT$_2$ is analogous.
We implement the generalized tetrad test based on the test statistic $T_n(\widehat{h}_0,\widehat{g}_0)$.

For each setting, we replicate 1000 simulations at sample size 500 and 1000, respectively.
Table~\ref{tab:2} summarizes the power of the tests. 
In Setting (I), all tests can  control the   type I error because both tetrad constraints are necessary conditions for $\H_0$ in the linear model.
However, the type I errors of the generalized tetrad tests based on $T_n$ are well below the nominal level of 0.05, 
as we have shown they are conservative due to the approximation to the asymptotic distribution of the test statistic using Chi-square distribution with one degree of freedom. 
In Setting (II.a), 
the classical tetrad test CT has  type I error much larger than the nominal  level $0.05$, 
and it increases as   the nonlinearity  becomes severer in Setting (II.b).
The generalized tetrad test GT$_1$ cannot control type I error either,
which is due to the model misspecification for the confounding bridge in this setting;
nonetheless,  the generalized tetrad test GT$_2$ can  control the type I error with a flexible working model for the estimation of the confounding bridge.
In Setting (III),  the power of all tests are close to unity  as the sample size increases to 1000.

In summary, in the linear models we consider in simulations, 
all tests can  control the empirical type I error when $\H_0$ is correct;
all tests have power close to unity as the sample size increases when $\H_0$ is incorrect.
In nonlinear models, the classical tetrad test can lead to large type I error even if $\H_0$ is correct;
in contrast, the generalized tetrad test GT$_2$ can still control the type I error.
When the parametric models for the confounding bridge functions are correctly specified, the power of GT$_1$ is higher than GT$_2$;
otherwise, the misspecification of the parametric model may result in large type I error of GT$_1$.
In Section~\ref{sec:simu-cov} of the Supplementary Material, we also evaluate the performance of the tetrad tests in the presence of observed covariates, and the results are similar.
In practice, we recommend implementing and comparing multiple applicable methods  for  testing $\H_0$, 
both the classical and generalized tetrad tests, both parametric and nonparametric working models,
to obtain a reliable  conclusion.

\begin{table}[tb]
    \centering 
    \caption{Power of tests at significance level 0.05. Columns in gray  correspond to  sample size 1000 and otherwise for 500}\label{tab:2}
    {
    \setlength{\aboverulesep}{0pt}
\setlength{\belowrulesep}{0pt}
    \begin{tabular}{ccc>{\columncolor{gray!40}}cc> {\columncolor{gray!40}}cc>{\columncolor{gray!40}}ccc}
\toprule 
\multicolumn{2}{c}{Settings} &\multicolumn{2}{c}{GT$_1$}& \multicolumn{2}{c}{GT$_2$} &\multicolumn{2}{c}{CT}  \\ 
\midrule
\multirow{3}{*}{$\H_0$ Correct} &(I)    &0.002 &0.006 &0.004 &0.004 &0.044 &0.046   \\ 
\cmidrule{2-8}
&(II.a) & 0.018 & 0.063 & 0.005 & 0.004 & 0.085 & 0.155  \\
&(II.b) &0.758 & 0.994 & 0.004 & 0.002 & 0.852 & 0.995 \\ 
\midrule
\multirow{2}{*}{$\H_0$ Incorrect} &(III.a) &0.546 & 0.906 & 0.448 & 0.865 & 0.734 & 0.976  \\ 
&(III.b) &0.999 & 1 & 0.974 & 1 & 1 & 1  \\
\bottomrule
\addlinespace[10pt]
\end{tabular}  
}
\end{table}

\section{Real data applications}\label{sec:application}
In this section, we illustrate the practical utility of the proposed generalized tetrad constraint and testing procedure through two real data applications. 
The first application, presented in Section~\ref{ssec:application-hs}, analyzes mental ability test scores from Holzinger and Swineford's study, investigating whether four spatial test scores share a common latent spatial ability factor. 
The second application, presented in Section~\ref{ssec:application-wvs}, examines moral attitudes towards dishonesty using data from the World Values Survey, testing whether responses to four questions about dishonest behaviors can be explained by a latent honesty factor. 
Both applications demonstrate how the generalized tetrad test can be used to validate latent variable models in different domains.

\subsection{Application to Holzinger and Swineford's data}
\label{ssec:application-hs}
We apply the proposed methods to  Holzinger and Swineford's ({HS}) data to analyze whether there exists a latent factor underlying students' performances in four tests about spatial ability.
The data set is available in R package ``\texttt{MBESS}''.
It contains scores of $n=301$ participants in 26 mental ability tests in the study on  human mental abilities conducted by \citet{holzinger1939study}.
The participants in the study are seventh- and eighth-grade students from two different schools, Pasteur School and Grant-White School.
The test scores are measurements about the participants' spatial, verbal, mental speed, memory, and mathematical ability. 
The data set also contains   covariates  gender ($1=$ female, $0=$ male), age (in months), grade ($1=$ eighth, $0=$ seventh), 
and the school that the participant is from ($1=$ Pasteur School, $0=$ Grant-White School), which are denoted by $\boldsymbol{V}$.
We consider the four scores that range in values from 2 to 51: visual perception test, cubes test,  paper form board test, and lozenges test, denoted by $(X,Y,Z,W)$, respectively. 
These four tests  are   designed to measure   spatial ability, denoted by $U$,
and we are interested in testing such an indicator model encoded in $\H^\prime_0: X\ind Y\ind Z\ind W\mid (U,\boldsymbol{V})$, that is, 
whether the association between these four test scores can be explained away by   a latent spatial ability  variable  after adjustment for observed covariates.

We apply  the classical   and the generalized tetrad tests to the data set.
For the classical tetrad test (CT), we obtain the test statistic based on  the two tetrad differences
$\sigma_{\widetilde{X}\widetilde{Y}}\sigma_{\widetilde{Z}\widetilde{W}}-\sigma_{\widetilde{X}\widetilde{W}}\sigma_{\widetilde{Y}\widetilde{Z}}=0$ and $\sigma_{\widetilde{X}\widetilde{Y}}\sigma_{\widetilde{Z}\widetilde{W}}-\sigma_{\widetilde{X}\widetilde{Z}}\sigma_{\widetilde{Y}\widetilde{W}}=0$, where $(\widetilde{X},\widetilde{Y},\widetilde{Z},\widetilde{W})$ are residuals after linear regression on $\boldsymbol{V}$.
The $p$-value  is calculated based on the asymptotic Chi-square distribution.
For the generalized tetrad test,
we model the confounding bridge function with an additive form: $h(w,\boldsymbol{v})=h_1(w)+h_2(\boldsymbol{v})$, where $h_2(\boldsymbol{v})$ is a linear function of $\boldsymbol{v}$.
We apply three methods to estimate  the confounding  bridge functions,
including GT$_1$: GMM with $h(w,\boldsymbol{v};\boldsymbol{\theta})=\theta_0+\theta_1w+\theta_2w^2+\boldsymbol{v}^{\T}\boldsymbol{\theta}_3$ and  $\boldsymbol{B}(z,\boldsymbol{v})=(1,z,z^2,\boldsymbol{v}^{\T})^{\T}$;
GT$_2$: the PSMD approach with $\boldsymbol{\xi}^{q_n}=\{\boldsymbol{\xi}^{q_{n,w}}(w)^\T,\boldsymbol{v}^\T\}^\T$, $\boldsymbol{\xi}^{q_{n,w}}(w)=\operatorname{Pol}(6)$, $\boldsymbol{p}^{k_n}=\{\boldsymbol{p}^{k_{n,z}}(z)^\T,\boldsymbol{v}^\T\}^\T$, $\boldsymbol{p}^{k_{n,z}}(z)=\operatorname{Pol}(9)$, $\operatorname{Pen}(h)=\Vert h\Vert_{L^2}^2$ and $\lambda_n=1\times10^{-5}$;
and GT$_3$: the PSMD approach with $\boldsymbol{\xi}^{q_n}=\{\boldsymbol{\xi}^{q_{n,w}}(w)^\T,\boldsymbol{v}^\T\}^\T$, $\boldsymbol{\xi}^{q_{n,w}}(w)=\operatorname{P-Spline}(3,2)$, $\boldsymbol{p}^{k_n}=\{\boldsymbol{p}^{k_{n,z}}(z)^\T,\boldsymbol{v}^\T\}^\T$, $\boldsymbol{p}^{k_{n,z}}(z)=\operatorname{P-Spline}(5,3)$, $\operatorname{Pen}(h)=\Vert h\Vert_{L^2}^2$ and $\lambda_n=1\times10^{-5}$, 
where $\operatorname{P-Spline}(r,k)$ denotes $r$th degree polynomial spline with $k$ quantile equally spaced knots;
see \citet{chen2015sieve} for the details of the PSMD method.
The estimation of $g_0(z)$ is analogous.
For each test, the $p$-values are calculated based on the Chi-square approximation.

Table~\ref{tab:4} shows the $p$-values of the tests.
All $p$-values exceed 0.2, 
then  none of the tests reject the null hypothesis.
In Section~\ref{ssec:add-hs} of the Supplementary Material, 
we  conduct additional generalized tetrad tests for different permutations of four variables $(X,Y,Z,W)$.
None of the testing results reject the null hypothesis at significance level 0.05.
This is empirical evidence in favor of that  $\H^\prime_0$  is a proper   indicator model for  the four  scores, 
and  the unobserved factor accounting for students' performances in the four tests can be interpreted as ``spatial ability''.
The concept of spatial ability  has long been  recognized in previous literature, such as, 
\citet{thurstone1938primary} represented spatial ability as a factor in his model of primary mental abilities.
The popular Cattell-Horn-Carroll (CHC) theory  about human cognitive abilities \citep{mcgrew2009chc,schneider2012cattell} also includes spatial ability as a factor  of the general intelligence.
\citet{meredith1977weighted} has applied   factor analysis to the HS data set and  shown that the four scores are indicators of the  spatial factor.
With the generalized tetrad tests, our analysis results reinforce  that the four scores are reasonable measurements   of the spatial ability.

\begin{table}[tb]
   \centering
   \caption{The $p$-values of tests in HS data}\label{tab:4}
{
\setlength{\aboverulesep}{0pt}
\setlength{\belowrulesep}{0pt}
\begin{tabular}{ccccc}
\toprule
\multirow{2}{*}{} &CT & GT$_1$ & GT$_2$ &GT$_3$\\
\midrule
$p$-values &0.255 & 0.417  & 0.504 & 0.477 \\
\bottomrule
\end{tabular}
}
\end{table}

\subsection{Application to the World Values Survey}
\label{ssec:application-wvs}

The World Values Survey wave 7 data (WVS Wave 7, 2017--2022) comprise surveys conducted in 64 countries and societies, which aim to assess the effect of stability or time-changing trend of values on the social, political and economic development of countries and societies.
The survey data are available online \citep{haerpfer2022world}.
We use the data collected   in  Canada to study whether there exists a latent factor underlying people's moral attitudes towards dishonesty.
The data set  contains complete responses of $n=3997$ participants to four questions about  the degree that the following statements can be justified: 
``claiming government benefits to which you are not entitled",
``avoiding a fare on public transport",
``cheating on taxes if you have a chance",
and ``someone accepting a bribe in the course of their duties".
These four questions are related to  the participants' attitudes towards various dishonest behaviors.
The responses to the questions, labeled as $X,Y,Z,W$, respectively, take values from $1$ to $10$. 
In addition, four observed covariates $\boldsymbol{V}$ are  available: 
gender ($1=$ female, $0=$ male), age (in years), highest educational level ($0$ to $8$) and income level ($1$ to $10$).

We are interested in testing such an indicator model encoded in $\H_0^\prime: X\ind Y\ind Z\ind W\mid (U,\boldsymbol{V})$, where $U$ stands for a latent variable about honesty, that is, 
whether these four responses are measurements of people's intrinsic honesty, and whether their associations can be explained away by   honesty  after adjustment for observed covariates.

We apply  the classical   and the generalized tetrad tests to the data set.
For the classical tetrad test (CT), we obtain the test statistic based on  the two tetrad differences
$\sigma_{\widetilde{X}\widetilde{Y}}\sigma_{\widetilde{Z}\widetilde{W}}-\sigma_{\widetilde{X}\widetilde{W}}\sigma_{\widetilde{Y}\widetilde{Z}}=0$ and $\sigma_{\widetilde{X}\widetilde{Y}}\sigma_{\widetilde{Z}\widetilde{W}}-\sigma_{\widetilde{X}\widetilde{Z}}\sigma_{\widetilde{Y}\widetilde{W}}=0$, where $(\widetilde{X},\widetilde{Y},\widetilde{Z},\widetilde{W})$ are residuals after linear regression on $\boldsymbol{V}$.
The $p$-value  is calculated based on the asymptotic Chi-square distribution.
For the generalized tetrad test,
we model the confounding bridge function $h_0$ with an additive form: $h(w,\boldsymbol{v})=h_1(w)+h_2(\boldsymbol{v})$, where $h_2(\boldsymbol{v})$ is a linear function of $\boldsymbol{v}$.
We apply three methods to estimate  the confounding  bridge functions,
including GT$_1$: GMM with $h(w,\boldsymbol{v};\boldsymbol{\theta})=\theta_0+\theta_1w+\theta_2w^2+\boldsymbol{v}^{\T}\boldsymbol{\theta}_3$ and $\boldsymbol{B}(z,\boldsymbol{v})=(1,z,z^2,\boldsymbol{v}^{\T})^{\T}$;
GT$_2$: the PSMD approach with $\boldsymbol{\xi}^{q_n}=\{\boldsymbol{\xi}^{q_{n,w}}(w)^\T,\boldsymbol{v}^\T\}^\T$, $\boldsymbol{\xi}^{q_{n,w}}(w)=\operatorname{Pol}(6)$, $\boldsymbol{p}^{k_n}=\{\boldsymbol{p}^{k_{n,z}}(z)^\T,\boldsymbol{v}^\T\}^\T$, $\boldsymbol{p}^{k_{n,z}}(z)=\operatorname{Pol}(9)$, $\operatorname{Pen}(h)=\Vert h\Vert_{L^2}^2$ and $\lambda_n=1\times10^{-5}$;
and GT$_3$: the PSMD approach with $\boldsymbol{\xi}^{q_n}=\{\boldsymbol{\xi}^{q_{n,w}}(w)^\T,\boldsymbol{v}^\T\}^\T$, $\boldsymbol{\xi}^{q_{n,w}}(w)=\operatorname{P-Spline}(3,2)$, $\boldsymbol{p}^{k_n}=\{\boldsymbol{p}^{k_{n,z}}(z)^\T,\boldsymbol{v}^\T\}^\T$, $\boldsymbol{p}^{k_{n,z}}(z)=\operatorname{P-Spline}(5,3)$, $\operatorname{Pen}(h)=\Vert h\Vert_{L^2}^2$ and $\lambda_n=1\times10^{-5}$, 
where $\operatorname{P-Spline}(r,k)$ denotes $r$th degree polynomial spline with $k$ quantile equally spaced knots;
see \citet{chen2015sieve} for the details of the PSMD method.
The estimation of $g_0(z)$ is analogous.
For each test, the $p$-values are calculated based on the Chi-square approximation.
Table~\ref{tab:5} shows the $p$-values of the tests.
All $p$-values exceed 0.2, 
then none of the tests reject the null hypothesis.

\begin{table}[tb]
   \centering
   \caption{The $p$-values of tests in the WVS data}   \label{tab:5}
   {
   \setlength{\aboverulesep}{0pt}
\setlength{\belowrulesep}{0pt}
    \begin{tabular}{ccccc}
\toprule
\multirow{2}{*}{} &CT &GT$_1$ & GT$_2$ &GT$_3$\\
\midrule
$p$-values&0.787 & 0.235  & 0.281  & 0.275 \\
\bottomrule
\end{tabular}
}
\end{table}

In Section~\ref{ssec:add-wvs} of the Supplementary Material, we conduct additional generalized tetrad  tests for different permutations of four variables $(X,Y,Z,W)$.
None of the testing results reject the null hypothesis at significance level 0.05.
This is empirical evidence in favor of that  $\H_0^\prime$  is a proper   indicator model for  the responses, and the unobserved factor accounting for participants' attitudes towards the four dishonest behaviors can be interpreted as honesty.
The factor honesty  has been recognized in previous literature, for example,  as a sixth factor of personality \citep{ashton2000honesty}, in addition to  the Big Five (Openness, Conscientiousness, Extraversion, Agreeableness, and Neuroticism).
Honesty literally means  being fair-minded, faithful and trustworthy, not deceitful, greedy or hypocritical
\citep{zeigler2020encyclopedia}.
In the WVS data set, the attitudes towards behaviors of cheating on welfare services and taxes, free riding and accepting a bribe are related to these aspects of honesty.
With the generalized tetrad tests, our analysis results reinforce that responses to the four questions are reasonable measurements of honesty.

\section{Discussion}
\label{sec:discussion}
In this paper, we have developed a novel tetrad constraint that generalizes   the classical tetrad constraint to nonlinear  and nonparametric models.
It is a necessary condition for conditional independence of four variables given a latent factor in both linear and nonlinear models.
The confounding bridge functions play a key role in our construction of the generalized tetrad constraint.
Based on the generalized tetrad constraint, we further propose a formal test for  conditional independence,  which  can control   type I error for significance level below 0.215, and has power approaching unity under certain conditions.
The generalized tetrad constraint has promising  applications in   causal discovery such as clustering observed variables that have a latent common cause in pure measurement models.
The generalized tetrad constraint allows for nonlinear structures in the causal graphical model,
and thus it complements and strengthens the classical tetrad constraint for  detecting latent cause of   observed variables  in causal graphical models. 
Besides,  for causal inference problems such as selecting and validating  negative control variables, the proposed approach can also be applied to extend previous tetrad-based  methods \citep[for example,][]{kummerfeld2024data}  to accommodate nonlinear models.

The proposed methods can be   extended in several directions.
First, it is of interest to extend the proposed approach to the multivariate setting where the observed variables and latent factor are  multiple and multi-dimensional.
Second, identification and estimation of the confounding bridge functions in the generalized tetrad constraint  rely on  the completeness assumption, which may not hold in practice,
for example when $Z$ and $W$ are both categorical and have different numbers of categories.
In this case, the confounding bridge functions may not be uniquely identified. 
This issue also arises  in several other causal inference problems,
and mitigation strategies   to the violation of completeness have been well established in, such as,  nonparametric instrumental variable problem \citep{santos2011instrumental}, nonignorable missing data \citep{li2023non} and proximal causal inference  \citep{zhang2023proximal}.
We refer to them for potential resolutions.
Third, the generalized tetrad  test is conservative due to the Chi-square approximation to the asymptotic distribution of the test statistic.
It is of interest to consider a more accurate approximation to improve the power of the test.
Besides, using a bootstrap procedure may offer a potential alternative to approximate the null distribution, which may improve finite-sample inference.
In addition, some causal discovery studies also concern   the relationship between the  detected latent variables,
and it is also of interest to adapt and apply  the generalized tetrad  constraint  in this situation.
We will consider these extensions somewhere else.

\bigskip
{

\paragraph{Supplementary Material}
The Supplementary Material includes proof of theorems, propositions, and useful lemmas, and additional simulations and real data analysis results.

}

\putbib

\end{bibunit}

\newpage

\appendix

\setcounter{assumption}{0}
\setcounter{lemma}{0}
\setcounter{table}{0}
\setcounter{figure}{0}
\setcounter{theorem}{0}
\setcounter{equation}{0}

\renewcommand{\theproposition}{S\arabic{proposition}}
\renewcommand{\thetheorem}{S\arabic{theorem}}
\renewcommand{\theassumption}{S\arabic{assumption}}
\renewcommand{\thesection}{S\arabic{section}}
\renewcommand{\theequation}{S\arabic{equation}}
\renewcommand{\thelemma}{S\arabic{lemma}} 
\renewcommand{\thetable}{S\arabic{table}}

{\centering \section*{Supplementary Material}}

\renewcommand {\thesection} {S\arabic{section}}
\renewcommand {\theexample} {S.\arabic{example}}
\renewcommand {\thelemma} {S.\arabic{lemma}}
\renewcommand {\theproposition} {S.\arabic{proposition}}
\renewcommand {\theassumption} {S.\arabic{assumption}}
\renewcommand {\thetheorem} {S.\arabic{theorem}}
\renewcommand {\theequation} {S.\arabic{equation}}
\renewcommand {\thetable} {S.\arabic{table}}
\renewcommand {\thefigure} {S.\arabic{figure}}
\newcommand*{\dt}[1]{\accentset{\mbox{\Large\bfseries .}}{#1}}

This supplement includes  proof of theorems, propositions, and useful lemmas, and additional simulations and real data analysis results.

\begin{bibunit}
\section{Proof of theorems, propositions, and useful lemmas}
\label{sec:proofs}
\subsection{Proof of Theorem~\ref{thm:1} }
\begin{proof}
For the first part of Theorem~\ref{thm:1},
Assumption~\ref{assum:1} guarantees the existence of at least one set of square-integrable functions $(h_0, g_0)$ satisfying \begin{equation}\label{eq:thm-1}
E(Y \mid Z) = E\{h_0(W) \mid Z\}, \quad E(Y \mid W) = E\{g_0(Z) \mid W\}.
\end{equation}
Assumption~\ref{assum:2} (completeness) ensures that these solutions are unique almost surely. 
Furthermore, by comparing Equations \eqref{eq:3}, \eqref{eq:4} with Equation \eqref{eq:thm-1}, it is evident that the functions satisfying the latter are identical to the confounding bridge functions defined in \eqref{eq:3} and \eqref{eq:4}. 
Since the equations in \eqref{eq:thm-1} depend only on the distribution of the observed variables $(Y, Z, W)$, the confounding bridge functions are thus identifiable from the observed data.

We now prove the second part of Theorem~\ref{thm:1}. 
If $\H_0$ holds, then $Y$ is conditionally independent of $X$ given $U$. 
From Assumption~\ref{assum:1}(i), we have
\[
E(Y \mid U, X) = E(Y \mid U) = E\{h_0(W) \mid U\} = E\{h_0(W) \mid U, X\}.
\]
Taking expectation with respect to the conditional distribution of $U$ given $X$ on both sides of the equation $E(Y \mid U, X) = E\{h_0(W) \mid U, X\}$, we obtain
\[
E(Y \mid X) = E\{h_0(W) \mid X\}, \quad \text{almost surely}.
\]
The same logic applies to $g_0(Z)$, yielding $E(Y \mid X) = E\{g_0(Z) \mid X\}$ almost surely.
This completes the proof of Equation \eqref{eq:6}.
\end{proof}

\subsection{Proof of Proposition~\ref{prop:2}}
\begin{proof}
If the integral equation \eqref{eq:3} has a solution $h_0(w)=\tau_0+\tau_1 w$, we have $\operatorname{cov}(Y,Z)=\operatorname{cov}\{h_0(W),Z\}=\operatorname{cov}(\tau_0+\tau_1 W,Z)$, then $\tau_1=\sigma_{YZ}\sigma_{ZW}^{-1}$, and $\operatorname{cov}\{h_0(W),X\}=\sigma_{XW}\sigma_{YZ}\sigma_{ZW}^{-1}$.
Then $\operatorname{cov}(Y,X)=\operatorname{cov}\{h_0(W),X\}$ is simplified to $\sigma_{XY}\sigma_{ZW}=\sigma_{XW}\sigma_{YZ}.$
Similar argument for $g_0(z)$.
\end{proof}

\subsection{Proof of Proposition~\ref{prop:3}}\label{ssec:proof-prop:3}
\begin{proof}
If $(X,Y,Z,W)$ are from a multivariate elliptical distribution, we have 
$$E(Y\mid Z)=\mu_Y+\Sigma_{YZ}\Sigma_{ZZ}^{-1}(Z-\mu_Z)=\mu_Y+\sigma_{YZ}\sigma_{ZZ}^{-1}(Z-\mu_Z),$$ and 
$$E(W\mid Z)=\mu_W+\Sigma_{WZ}\Sigma_{ZZ}^{-1}(Z-\mu_Z)=\mu_W+\sigma_{WZ}\sigma_{ZZ}^{-1}(Z-\mu_Z),$$ 
where $\Sigma_{YZ},\Sigma_{ZZ},\Sigma_{WZ}$ are the corresponding elements of the matrix $\boldsymbol{\Sigma}.$
By the completeness condition, the solution to \eqref{eq:3} is exactly the unique solution to \eqref{eq:5}, which is in the form $h_0(w)=\tau_0+\tau_1 w$ with $\tau_1=\sigma_{YZ}\sigma_{ZW}^{-1}$.
Note that $E(Y\mid X)=c_0+\sigma_{XY}\sigma_{XX}^{-1}X$,
and $E\{h_0(W)\mid X\}=\tau_0+\tau_1 E(W\mid X)=d_0+\sigma_{YZ}\sigma_{ZW}^{-1}\sigma_{XW}\sigma_{XX}^{-1}X$, where $c_0,d_0$ are some constants.
Then $E(Y\mid X)=E\{h_0(W)\mid X\}$ is equivalent to $\sigma_{XY}\sigma_{ZW}=\sigma_{XW}\sigma_{YZ}.$
Similar argument for $g_0(z)$.
\end{proof}

\subsection{Proof of Proposition~\ref{prop:4}}

\begin{proof}
The proof follows directly from \citet[][Theorem 1]{shao2014martingale}.
\end{proof}

\subsection{Proof of Theorem~\ref{thm:5}}
\subsubsection{Proof of Theorem~\ref{thm:5}(a)}
\begin{proof}
The consistency of the AMGT estimator $\operatorname{AMGT}_n(\widehat{h}_0,\widehat{g}_0)$ follows from the consistency of both $\operatorname{MGT}_n(\widehat{h}_0)$ and $\operatorname{MGT}_n(\widehat{g}_0)$.
Note that we just need to prove that $\operatorname{MGT}_n(\widehat{h}_0)\to\operatorname{MGT}(h_0)$ in probability, and the proof of $\operatorname{MGT}_n(\widehat{g}_0)\to\operatorname{MGT}(g_0)$ is analogous.
We have
$$\operatorname{MGT}_n(\widehat{h}_0)^2-\operatorname{MGT}(h_0)^2=\left\{\operatorname{MGT}_n(\widehat{h}_0)^2-\operatorname{MGT}(\widehat{h}_0)^2\right\}+\left\{\operatorname{MGT}(\widehat{h}_0)^2-\operatorname{MGT}(h_0)^2\right\}.$$ 
The first term $\operatorname{MGT}_n(\widehat{h}_0)^2-\operatorname{MGT}(\widehat{h}_0)^2\to0$ as $n\to\infty$ almost surely,
which follows along similar steps as in the proof of Theorem 3 in \citet{shao2014martingale}.
The second term $\operatorname{MGT}(\widehat{h}_0)^2-\operatorname{MGT}(h_0)^2\to0$ in probability, because
\begin{align}
\bigl\vert&\operatorname{MGT}(\widehat{h}_0)^2-\operatorname{MGT}(h_0)^2\bigr\vert\nonumber\\\nonumber&=\biggl\vert E\Bigl\{\Bigl(\bigl[R(\widehat{h}_0)-E\{R(\widehat{h}_0)\}\bigr]\bigl[R^\prime(\widehat{h}_0)-E\{R(\widehat{h}_0)\}\bigr]-R(h_0)R^\prime(h_0)\Bigr)\vert X-X^\prime\vert\Bigr\}\biggr\vert\\  
\nonumber&= \begin{vmatrix}
E\begin{Bmatrix}
\begin{pmatrix}
\bigl[R(\widehat{h}_0)-E\{R(\widehat{h}_0)\}-R({h}_0)\bigr] R^\prime({h}_0) \\+
\bigl[R^\prime(\widehat{h}_0)-E\{R^\prime(\widehat{h}_0)\}-R^\prime({h}_0)\bigr] R({h}_0) \\+
\bigl[R(\widehat{h}_0)-E\{R(\widehat{h}_0)\}-R({h}_0)\bigr]\bigl[R^\prime(\widehat{h}_0)-E\{R(\widehat{h}_0)\}-R^\prime({h}_0)\bigr]
\end{pmatrix}\vert X-X^\prime\vert
\end{Bmatrix}
\end{vmatrix}\\  
\nonumber&\leq\biggl\vert E\Bigl(\bigl[\Tilde{h}(W)-E\{\Tilde{h}(W)\}\bigr]R^\prime(h_0)\vert X-X^\prime\vert\Bigr)\biggr\vert\\
\nonumber&\quad+\biggl\vert E\Bigl(\bigl[\Tilde{h}(W^\prime)-E\{\Tilde{h}(W)\}\bigr]R(h_0)\vert X-X^\prime\vert\Bigr)\biggr\vert\\
\nonumber&\quad+\biggl\vert E\Bigl(\bigl[\Tilde{h}(W)-E\{\Tilde{h}(W)\}\bigr]\bigl[\Tilde{h}(W^\prime)-E\{\Tilde{h}(W)\}\bigr]\vert X-X^\prime\vert\Bigr)\biggr\vert\\
\nonumber&\leq2\Vert \widehat{h}_0-h_0\Vert_{\infty}E\bigl[\bigl\{\vert R(h_0)\vert+\vert R^\prime(h_0)\vert\bigr\}\vert X-X^\prime\vert\bigr]+4\Vert \widehat{h}_0-h_0\Vert_{\infty}^2E(\vert X-X^\prime\vert)\\
\nonumber&\leq4\Vert \widehat{h}_0-h_0\Vert_{\infty}E\{\vert Y-h_0(W)\vert^2+\vert X\vert^2\}+8\Vert \widehat{h}_0-h_0\Vert_{\infty}^2E(\vert X\vert)\\
\nonumber&=o_p(1),
\end{align}
where $\Tilde{h}=\widehat{h}_0-h_0,$ and the last equality follows from Assumptions~\ref{assum:4} and~\ref{assum:5}(i).
\end{proof}

\subsubsection{Proof of Theorem~\ref{thm:5}(b)}\label{ssec:aym-mgt}
We first prove the following lemma for the asymptotic distribution of $n\operatorname{MGT}_n(\widehat{h}_0)^2$ under $\H_0$.
\begin{lemma}\label{lem:mgt}
Under Assumptions~\ref{assum:1}--\ref{assum:5}, if $\H_0$ is correct, then
$n\operatorname{MGT}_n(\widehat{h}_0)^2$ converges in distribution to a random variable, 
$\pi^{-1}\int_{s\in\mathbb{R}}\vert\Gamma(s)\vert^2s^{-2}\dif s$,
where  $\Gamma(\cdot)$ is a complex-valued mean zero Gaussian random process with covariance function: 
\begin{eqnarray*}
\mathrm{cov}_{\Gamma}(s_1,s_2)&=&E\{\zeta_{s_1}(O)\overline{\zeta_{s_2}(O)}\},\\
\zeta_s(O)&=&R(h_0)e^{isX}-R(h_0)E(e^{isX})+\varphi_s(O;h_0), s,s_1,s_2\in\mathbb{R}.
\end{eqnarray*}
\end{lemma}
\begin{proof}
Under $\H_0,$ we have $\operatorname{MGT}(h_0)=0,$ and $\Lambda_s(h_0)=0$ almost surely.
    
Denote $\Vert\Gamma\Vert_q^2=\pi^{-1}\int_{s\in\mathbb{R}} \vert\Gamma(s)\vert^2s^{-2}\dif s$.
Similar to the proof of Theorem 4 in \citet{shao2014martingale}, 
we can construct a sequence of random variables $\{Q_n(\delta)\}$, such that

{(I) $Q_n(\delta)\to Q(\delta)$ in distribution for each $\delta>0$; } 

(II) for all $\varepsilon>0$, $\lim_{\delta\to0}\varlimsup_{n\to\infty}\operatorname{pr}\bigl\{\vert Q_n(\delta)-\Vert n^{1/2}\widehat{\Lambda}_s(\widehat{h}_0)\Vert_q^2\vert>\varepsilon\bigr\}=0$;

(III) $E\vert Q(\delta)-\Vert\Gamma\Vert_q^2\vert\to0$ as $\delta\to0.$

Then the weak convergence of $\Vert n^{1/2}\widehat{\Lambda}_s(\widehat{h}_0)\Vert_q^2$ 
to $\Vert\Gamma\Vert_q^2$ follows from Theorem 8.6.2 of \citet{resnick2014probability}.

We define the region $D(\delta) = \{s:\delta \leq \vert s\vert \leq 1/\delta\}$ for each $\delta > 0$.
For a process $\xi(s)$, denote $\Vert\xi\Vert_{q(\delta)}^2=\pi^{-1}\int_{D(\delta)} \vert\xi(s)\vert^2s^{-2}\dif s$. Let $Q_n(\delta)=\Vert n^{1/2}\widehat{\Lambda}_s(\widehat{h}_0)\Vert_{q(\delta)}^2$ and $Q(\delta)=\Vert\Gamma\Vert_{q(\delta)}^2$.
We now verify (I)--(III) to obtain the conclusion in Theorem~\ref{thm:5}(b).

For the proof of (I), let $\Delta_s(\widehat{h}_0)=\widehat{\Lambda}_s(\widehat{h}_0)-\widehat{\Lambda}_s(h_0)-\Lambda_s(\widehat{h}_0)+\Lambda_s(h_0)$.
Note that if $\Lambda_s(h_0)=0$, then
\begin{equation}
\label{eq:s3-1}
\begin{aligned}
n^{1/2}\widehat{\Lambda}_s(\widehat{h}_0)=n^{1/2}\{\widehat{\Lambda}_s(h_0)-\Lambda_s(h_0)\}+n^{1/2}\{\Lambda_s(\widehat{h}_0)-\Lambda_s(h_0)\}+n^{1/2}\bigl\{\Delta_s(\widehat{h}_0)\bigr\}.
\end{aligned}
\end{equation}
Let $l_s(o;h)=\{y-h(w)\}e^{isx}-\{y-h(w)\}E(e^{isX})-E\{Y-h(W)\}e^{isx}$, and $L_{s,n}(h)=n^{-1}\sum_{j=1}^n\bigl[l_s(O_j;h)-E\{l_s(O;h)\}\bigr],$ then we have
\(\Vert n^{1/2}\{L_{s,n}(\widehat{h}_0)-L_{s,n}(h_0)\}\Vert_{q(\delta)}^2=o_p(1)\) follows from the Donsker condition, using the empirical process theory \citep{van1996weak}.
For the first term $n^{1/2}\{\widehat{\Lambda}_s(h_0)-\Lambda_s(h_0)\}$ in Equation~\eqref{eq:s3-1},
we have
\begin{align*}
\widehat{\Lambda}_s(h_0)-\Lambda_s(h_0)
=\frac{1}{n}\sum_{j=1}^nR_j(h_0)e^{isX_j}-\frac{1}{n^2}\sum_{j,k=1}^n\biggl\{\frac{1}{2}R_j(h_0)e^{isX_k}+\frac{1}{2}R_k(h_0)e^{isX_j}\biggr\}, 
\end{align*}
we denote $T_n(\boldsymbol{O}_j,\boldsymbol{O}_k)=\bigl\{R_j(h_0)e^{isX_k}+R_k(h_0)e^{isX_j}\bigr\}/2$, and $T_{1n}(\boldsymbol{O}_j)=E\bigl\{T_n(\boldsymbol{O}_j,\boldsymbol{O}_k)\mid \boldsymbol{O}_j\bigr\}(j\neq k)$, then we have 
\begin{equation}
\label{eq:s3-2}
\begin{aligned}
n^{1/2}\bigl\{\widehat{\Lambda}_s(h_0)-\Lambda_s(h_0)\bigr\}&=n^{-1/2}\sum_{j=1}^nR_j(h_0)e^{isX_j}-2n^{-1/2}\sum_{j=1}^nT_{1n}(\boldsymbol{O}_j)+t_{1,n}(s)\\
&=n^{1/2}L_{s,n}(h_0)+t_{1,n}(s),
\end{aligned}
\end{equation}
where $\Vert t_{1,n}\Vert_{q(\delta)}^2=o_p(1)$, which follows from the result of U- and V-statistics, see, for example, Chapters 3.2 and 3.5 in \citet{shao2003mathematical}. 
Similarly, for the third term $n^{1/2}\bigl\{\Delta_s(\widehat{h}_0)\bigr\}$ in Equation~\eqref{eq:s3-1},
\begin{equation}
\label{eq:s3-3}
\begin{aligned}
n^{1/2}\bigl\{\Delta_s(\widehat{h}_0)\bigr\}=n^{1/2}\bigl\{L_{s,n}(\widehat{h}_0)-L_{s,n}(h_0)\bigr\}+t_{2,n}(s)=r_{2,n}(s),
\end{aligned}
\end{equation}
where $\Vert t_{2,n}\Vert_{q(\delta)}^2=o_p(1)$, and $\Vert r_{2,n}\Vert_{q(\delta)}^2=o_p(1)$ follows from the Donsker condition.

We denote $A_n(s)=n^{1/2}L_{s,n}(h_0)$ and $B_n(s)=n^{-1/2}\sum_{j=1}^n{ \varphi_s}(\boldsymbol{O}_j;h_0).$
We have \[c_1(s,s)=E\vert A_n(s)\vert^2=E\bigl\vert  R(h_0)e^{isX}-R(h_0)E(e^{isX} )\bigr\vert^2<\infty,\] and $E\vert B_n(s)\vert^2=c_2(s,s)<\infty$,
by Assumptions~\ref{assum:4} and~\ref{assum:5}.

Let $\Gamma_n(s)=A_n(s)+B_n(s)$,
combined Equations~\eqref{eq:s3-1}--\eqref{eq:s3-3} with Assumption~\ref{assum:5}(ii), we have
\begin{equation}
\label{eq:s4}
\begin{aligned}
n^{1/2}\widehat{\Lambda}_s(\widehat{h}_0)
&=n^{1/2}\{\widehat{\Lambda}_s(h_0)-\Lambda_s(h_0)\}+n^{-1/2}\sum_{j=1}^n{  \varphi_s}(\boldsymbol{O}_j;h_0)+r_{1,n}(s)+n^{1/2}\bigl\{\Delta_s(\widehat{h}_0)\bigr\}\\
&= \Gamma_n(s)+r_n(s),
\end{aligned}
\end{equation}
where $r_n(s)=r_{1,n}(s)+r_{2,n}(s)+t_{1,n}(s)$, with $r_{1,n}(s)$ defined in Assumption~\ref{assum:5}(ii), and $\Vert r_n\Vert_{q(\delta)}^2=o_p(1).$

Let $\widetilde{Q}_n(\delta)=\pi^{-1}\int_{D(\delta)} \vert \Gamma_n(s)\vert^2s^{-2}\dif s$. 
Therefore, \eqref{eq:s4} implies that  $\Vert n^{1/2}\widehat{\Lambda}_s(\widehat{h}_0)-\Gamma_n\Vert_{q(\delta)}^2=\Vert r_n\Vert_{q(\delta)}^2=o_p(1)$, and
\begin{align}
\vert Q_n(\delta)-\widetilde{Q}_n(\delta)\vert\nonumber&\leq\pi^{-1}\int_{D(\delta)} \Bigl\vert\bigl\vert n^{1/2}\widehat{\Lambda}_s(\widehat{h}_0)\bigr\vert^2- \vert \Gamma_n(s)\vert^2\Bigr\vert s^{-2}\dif s\\
\nonumber&=\pi^{-1}\int_{D(\delta)} \Bigl\vert \bigl\{n^{1/2}\widehat{\Lambda}_s(\widehat{h}_0)- \Gamma_n(s)\bigr\}\overline{\Gamma_n(s)}+n^{1/2}\widehat{\Lambda}_s(\widehat{h}_0)\bigl\{n^{1/2}\overline{\widehat{\Lambda}_s(\widehat{h}_0)}- \overline{\Gamma_n(s)}\bigr\}\Bigr\vert s^{-2}\dif s\\
\nonumber&\leq\Vert n^{1/2}\widehat{\Lambda}_s(\widehat{h}_0)-\Gamma_n\Vert_{q(\delta)}\bigl\{\Vert n^{1/2}\widehat{\Lambda}_s(\widehat{h}_0)\Vert_{q(\delta)}+\Vert\Gamma_n\Vert_{q(\delta)}\bigr\}\\
\nonumber&\leq\Vert n^{1/2}\widehat{\Lambda}_s(\widehat{h}_0)-\Gamma_n\Vert_{q(\delta)}\bigl\{2\Vert n^{1/2}\widehat{\Lambda}_s(\widehat{h}_0)-\Gamma_n\Vert_{q(\delta)}+3\Vert\Gamma_n\Vert_{q(\delta)}\bigr\}\\
\nonumber&=o_p(\widetilde{Q}_n(\delta)^{1/2})=o_p(1),
\end{align}
where the second inequality is from the Cauchy-Schwarz inequality.

We then consider the asymptotic distribution of $\widetilde{Q}_n(\delta)$.

By Assumptions~\ref{assum:4},~\ref{assum:5} and the multivariate form of the Lindeberg-Feller central limit theorem,
every finite sequence of $\{\Gamma_n(s)\}_{s\in\mathbb{R}}$ converges to that of $\{\Gamma(s)\}_{s\in\mathbb{R}}$ in distribution.  

Given a positive integer $m$, choose a partition $\{D_k\}_{k=1}^M$ of $D(\delta)$ into $M=M(m)$ measurable sets with diameter at most $1/m$.
Define 
$$Q_n^m(\delta)=\sum_{k=1}^M\pi^{-1}\int_{D_k} \vert\Gamma_n(s_0(k))\vert^2s^{-2}\dif s, \quad\text{and}\quad Q^m(\delta)=\sum_{k=1}^M\pi^{-1}\int_{D_k} \vert\Gamma(s_0(k))\vert^2s^{-2}\dif s,$$
where $\{s_0(k)\}_{k=1}^M$ are distinct points such that $s_0(k)\in D_k$.  

By the continuous mapping theorem, $Q_n^m(\delta)\to Q^m(\delta)$ in distribution for any positive integer $m$. 
By again applying Theorem 8.6.2 of \citet{resnick2014probability}, if we can additionally show that 
$$\text{(i)} \varlimsup_{m\to\infty}\varlimsup_{n\to\infty}E\vert Q_n^m(\delta)-\widetilde{Q}_n(\delta)\vert=0,\quad \text{and} \quad
\text{(ii)} \varlimsup_{m\to\infty}E\vert Q^m(\delta)-Q(\delta)\vert=0,$$
then $\widetilde{Q}_n(\delta)\to Q(\delta)$ in distribution.

We first prove (i). Let $\beta_n(m)=\sup_{s_1,s_2}E\bigl\vert \vert\Gamma_n(s_1)\vert^2-\vert\Gamma_n(s_2)\vert^2\bigr\vert$ under the restrictions $\delta<\vert s_1\vert,\vert s_2\vert<1/\delta$ and $\vert s_1-s_2\vert<1/m$. 
To show that $\varlimsup_{m\to\infty}\varlimsup_{n\to\infty}E\vert Q_n^m(\delta)-\widetilde{Q}_n(\delta)\vert=0$, 
we have
\begin{equation}
\begin{aligned}
E\vert Q_n^m(\delta)-\widetilde{Q}_n(\delta)\vert\nonumber&=\pi^{-1}E\biggl\vert \int_{D(\delta)} \vert\Gamma_n(s)\vert^2s^{-2}\dif s-\sum_{k=1}^M\int_{D_k} \vert\Gamma_n(s_0(k))\vert^2s^{-2}\dif s\biggr\vert\\
\nonumber&=\pi^{-1}E\biggl\vert\sum_{k=1}^M\int_{D_k} \bigl\{\vert\Gamma_n(s)\vert^2-\vert\Gamma_n(s_0(k))\vert^2\bigr\}s^{-2}\dif s\biggr\vert\\
\nonumber&\leq\pi^{-1}\beta_n(m)\int_{D(\delta)}s^{-2} \dif s.
\end{aligned}
\end{equation}
We only need to show $\varlimsup_{m\to\infty}\varlimsup_{n\to\infty}\beta_n(m)=0$. 
By applying
the Cauchy-Schwarz inequality, we have
\begin{align}\label{eq:s5}
\beta_n(m)\nonumber&=\sup_{s_1,s_2}E\bigl\vert\{\Gamma_n(s_1)-\Gamma_n(s_2)\}\overline{\Gamma_n(s_1)}+\Gamma_n(s_2)\{\overline{\Gamma_n(s_1)}-\overline{\Gamma_n(s_2)}\}\bigr\vert\\
\nonumber&\leq\sup_{s_1,s_2}\bigl(\{E\vert\Gamma_n(s_1)-\Gamma_n(s_2)\vert^2\}^{1/2}\bigl[\{E\vert\Gamma_n(s_1)\vert^2\}^{1/2}+\{E\vert\Gamma_n(s_2)\vert^2\}^{1/2}\bigr]\bigr)\\
\nonumber&\leq 2\sup_{s_1,s_2}\bigl\vert \operatorname{cov}_{\Gamma}(s_1,s_1)-\operatorname{cov}_{\Gamma}(s_1,s_2)-\operatorname{cov}_{\Gamma}(s_2,s_1)+\operatorname{cov}_{\Gamma}(s_2,s_2)\bigr\vert^{1/2}\\
\nonumber&\quad\bigl[\{E\vert A_n(s_1)\vert^2\}^{1/2}+\{E\vert A_n(s_2)\vert^2\}^{1/2}+\{E\vert B_n(s_1)\vert^2\}^{1/2}+\{E\vert B_n(s_2)\vert^2\}^{1/2}\bigr]\bigr)\\
&\leq C \sup_{s_1,s_2}\bigl\vert \operatorname{cov}_{\Gamma}(s_1,s_1)-\operatorname{cov}_{\Gamma}(s_1,s_2)-\operatorname{cov}_{\Gamma}(s_2,s_1)+\operatorname{cov}_{\Gamma}(s_2,s_2)\bigr\vert^{1/2}.
\end{align}
Then \(\varlimsup_{m\to\infty}\varlimsup_{n\to\infty}\beta_n(m)\nonumber=0\) follows from the uniform continuity of $\{\operatorname{cov}_{\Gamma}(s_1,s_2)\}.$
According to the above steps, we have $\varlimsup_{m\to\infty}\varlimsup_{n\to\infty}E\vert Q_n^m(\delta)-\widetilde{Q}_n(\delta)\vert=0$, and (i) holds.
Similarly, we have $\varlimsup_{m\to\infty}E\vert Q^m(\delta)-Q(\delta)\vert=0$, and (ii) holds.
Therefore, $\widetilde{Q}_n(\delta)\to Q(\delta)$ in distribution,
and (I) holds because we have proved that $Q_n(\delta)-\widetilde{Q}_n(\delta)=o_p(\widetilde{Q}_n(\delta)^{1/2})=o_p(1)$.

For the proof of (II), we have
\begin{align}
\vert Q_n&(\delta)-\Vert n^{1/2}\widehat{\Lambda}_s(\widehat{h}_0)\Vert_q^2\vert\nonumber\\\nonumber&=\pi^{-1}\int_{\vert s\vert<\delta} \vert n^{1/2}\widehat{\Lambda}_s(\widehat{h}_0)\vert^2s^{-2}\dif s+\pi^{-1}\int_{\vert s\vert>1/\delta} \vert n^{1/2}\widehat{\Lambda}_s(\widehat{h}_0)\vert^2s^{-2}\dif s\\
\nonumber&\leq C\Bigl[\int_{\vert s\vert<\delta} \vert n^{1/2}\bigl\{\widehat{\Lambda}_s(h_0)-\Lambda_s(h_0)\bigr\}\vert^2s^{-2}\dif s+\int_{\vert s\vert>1/\delta} \vert n^{1/2}\bigl\{\widehat{\Lambda}_s(h_0)-\Lambda_s(h_0)\bigr\}\vert^2s^{-2}\dif s\Bigr]\\
\nonumber&\quad+C\Bigl[\int_{\vert s\vert<\delta} \vert n^{1/2}\bigl\{\widehat{\Lambda}_s(\widehat{h}_0)-\widehat{\Lambda}_s(h_0)\bigr\}\vert^2s^{-2}\dif s +\int_{\vert s\vert>1/\delta} \vert n^{1/2}\bigl\{\widehat{\Lambda}_s(\widehat{h}_0)-\widehat{\Lambda}_s(h_0)\bigr\}\vert^2s^{-2}\dif s\Bigr].
\end{align}

By the proof of Theorems 3 and 4 in \citet{shao2014martingale}, we can derive that for all $\varepsilon>0$, $$\lim_{\delta\to0}\varlimsup_{n\to\infty}\operatorname{pr}\Bigl[\int_{\vert s\vert<\delta} \vert n^{1/2}\bigl\{\widehat{\Lambda}_s(h_0)-\Lambda_s(h_0)\bigr\}\vert^2s^{-2}\dif s+\int_{\vert s\vert>1/\delta} \vert n^{1/2}\bigl\{\widehat{\Lambda}_s(h_0)-\Lambda_s(h_0)\bigr\}\vert^2s^{-2}\dif s>\varepsilon\Bigr]=0.$$ 

We next prove that for all $\varepsilon>0$,  $$\lim_{\delta\to0}\varlimsup\limits_{n\to\infty}\operatorname{pr}\Bigl[\int_{\vert s\vert<\delta} \vert n^{1/2}\bigl\{\widehat{\Lambda}_s(\widehat{h}_0)-\widehat{\Lambda}_s(h_0)\bigr\}\vert^2s^{-2}\dif s+\int_{\vert s\vert>1/\delta} \vert n^{1/2}\bigl\{\widehat{\Lambda}_s(\widehat{h}_0)-\widehat{\Lambda}_s(h_0)\bigr\}\vert^2s^{-2}\dif s>\varepsilon\Bigr]=0.$$
Note that
\begin{eqnarray*}
\vert n^{1/2}\bigl\{\widehat{\Lambda}_s(\widehat{h}_0)-\widehat{\Lambda}_s(h_0)\bigr\}\vert^2
&=&n\biggl\vert n^{-1}\sum_{j=1}^n\bigl\{\widehat{h}_0(W_j)-h_0(W_j)\bigr\}n^{-1}\sum_{j=1}^n\bigl\{e^{isX_j}-E(e^{isX})\bigr\}\\
&&-n^{-1}\sum_{j=1}^n\bigl\{\widehat{h}_0(W_j)-h_0(W_j)\bigr\}\bigl\{e^{isX_j}-E(e^{isX})\bigr\}\biggr\vert^2\\
&\leq& 4n^{-1}\sum_{j=1}^n\bigl\{\widehat{h}_0(W_j)-h_0(W_j)\bigr\}^2\sum_{j=1}^n\bigl\vert e^{isX_j}-E(e^{isX})\bigr\vert^2
\end{eqnarray*}
by the Cauchy-Schwarz inequality, then we have
\begin{eqnarray*}
\int_{\vert s\vert<\delta} \vert n^{1/2}\bigl\{\widehat{\Lambda}_s(\widehat{h}_0)-\widehat{\Lambda}_s(h_0)\bigr\}\vert^2s^{-2}\dif s
&\leq&4n^{-1}\int_{\vert s\vert<\delta}\sum_{j=1}^n\bigl\vert e^{isX_j}-E(e^{isX})\bigr\vert^2s^{-2}\dif s\sum_{j=1}^n\bigl\{\widehat{h}_0(W_j)-h_0(W_j)\bigr\}^2\\
&\leq&8\biggl[n^{-1/2}\sum_{j=1}^nE\bigl\{\vert X_j-X\vert G(\vert X_j-X\vert\delta)\mid X_j\bigr\}\biggr]\\
&&\cdot\biggl[n^{-1/2}\sum_{j=1}^n\{\widehat{h}_0(W_j)-h_0(W_j)\bigr\}^2\biggr],
\end{eqnarray*}
where the function $G(t)=\int_{\vert z\vert<t}(1-\cos z)/z^2\dif z$ in the last inequality, and we have used the fact in \citet{szekely2007measuring}.
Note that
\[n^{-1/2}\sum_{j=1}^nE\bigl\{\vert X_j-X\vert G(\vert X_j-X\vert\delta)\mid X_j\bigr\}\] is $O_p(1)$ when $\delta$ is fixed and $n\to\infty$, and 
$\varlimsup\limits_{n\to\infty}\operatorname{pr}\bigl[n^{-1/2}\sum_{j=1}^nE\bigl\{\vert X_j-X\vert G(\vert X_j-X\vert\delta)\mid X_j\bigr\}>\varepsilon\bigr]\to0$ as $\delta\to0$.
For the other term that is free of $\delta$, we have
\begin{eqnarray*}
n^{-1/2}\sum_{j=1}^n\{\widehat{h}_0(W_j)-h_0(W_j)\bigr\}^2&=&n^{1/2}\biggl[\frac{1}{n}\sum_{j=1}^n\{\widehat{h}_0(W_j)-h_0(W_j)\bigr\}^2-E\{\widehat{h}_0(W)-h_0(W)\bigr\}^2\biggr]\\
&&+n^{1/2}E\{\widehat{h}_0(W)-h_0(W)\bigr\}^2,
\end{eqnarray*}
which is $O_p(1)$ by Assumptions~\ref{assum:4} (Donsker condition) and~\ref{assum:5}(i).
Then for all $\varepsilon>0$,  $$\lim_{\delta\to0}\varlimsup\limits_{n\to\infty}\operatorname{pr}\Bigl[\int_{\vert s\vert<\delta}\vert n^{1/2}\bigl\{\widehat{\Lambda}_s(\widehat{h}_0)-\widehat{\Lambda}_s(h_0)\bigr\}\vert^2s^{-2}\dif s >\varepsilon\Bigr]=0.$$
Analogously, we can prove that $\lim_{\delta\to0}\varlimsup_{n\to\infty}\operatorname{pr}\bigl[\int_{\vert s\vert>1/\delta} \vert n^{1/2}\{\widehat{\Lambda}_s(\widehat{h}_0)-\widehat{\Lambda}_s(h_0)\}\vert^2s^{-2}\dif s>\varepsilon\bigr]=0$. 
Therefore,
(II) holds.

We finally prove (III). By simple calculations, we have, $$E\vert Q(\delta)-\Vert\Gamma\Vert_q^2\vert=\pi^{-1}\int_{\vert s\vert<\delta} \operatorname{cov}_{\Gamma}(s,s)s^{-2}\dif s+\pi^{-1}\int_{\vert s\vert>1/\delta} \operatorname{cov}_{\Gamma}(s,s)s^{-2}\dif s .$$
By Assumptions~\ref{assum:3} and~\ref{assum:4}, $\pi^{-1}\int_{s\in\mathbb{R}}c_1(s,s)s^{-2}\dif s=2E\{R(h_0)^2\vert X-X^\prime\vert\}-E\{R(h_0)^2\}E\vert X-X^\prime\vert<\infty$.
By Assumption~\ref{assum:5}(ii), $\pi^{-1}\int_{s\in\mathbb{R}}c_2(s,s)s^{-2}\dif s<\infty$.
Note that $\operatorname{cov}_{\Gamma}(s,s)\leq 2c_1(s,s)+2c_2(s,s),$
then $E\vert Q(\delta)-\Vert\Gamma\Vert_q^2\vert\to0$ as $\delta\to0$, that is, (III) holds.

Therefore, $n\operatorname{MGT}_n(\widehat{h}_0)^2\to\Vert\Gamma\Vert_q^2$ in distribution under $\H_0$.
\end{proof}
We now prove Theorem~\ref{thm:5}(b).
\begin{proof}
To avoid possible confusion, we use $\Gamma_h$ to denote the Gaussian process involved in the asymptotic distribution of $n\operatorname{MGT}_n(\widehat{h}_0)^2$.
By Assumption~\ref{assum:6}, $S_n(\widehat{h}_0)$ converges to $\pi^{-1}\int_{s\in\mathbb{R}}\operatorname{cov}_{\Gamma}(s,s)s^{-2}\dif s$ in probability under $\H_0$.
Analogously, under a set of conditions parallel to Assumptions~\ref{assum:3}--\ref{assum:6},
$n\operatorname{MGT}_n(\widehat{g}_0)^2$ converges in distribution to $\Vert\Gamma_g\Vert_q^2$ under $\H_0$, which is a weighted integral of the norms of a complex Gaussian process $\Gamma_g$,
and $S_n(\widehat{g}_0)$ converges to $E\Vert\Gamma_g\Vert_q^2$ in probability under $\H_0$.
The above proof can also be extended to the derivation of the asymptotic distribution of $\operatorname{AMGT}_n(\widehat{h}_0,\widehat{g}_0)^2$, which incorporates the residual functions for $\widehat{h}_0$ and $\widehat{g}_0$ with $\{\widehat{\Lambda}_s(\widehat{h}_0),\widehat{\Lambda}_s(\widehat{g}_0)\}^\T$.
Using similar arguments,
we obtain that $n\operatorname{AMGT}_n(\widehat{h}_0,\widehat{g}_0)^2$ converges in distribution to $\Vert\Gamma_{g,h}\Vert_q^2$ under $\H_0$, which is a weighted integral of the norms of a complex Gaussian process $\Gamma_{g,h}$.
Since $\operatorname{AMGT}_n(\widehat{h}_0,\widehat{g}_0)^2=\operatorname{MGT}_n(\widehat{h}_0)^2+\operatorname{MGT}_n(\widehat{g}_0)^2$, we have $E\Vert\Gamma_{h,g}\Vert_q^2=E\Vert\Gamma_h\Vert_q^2+E\Vert\Gamma_g\Vert_q^2$,
and $S_n=S_n(\widehat{h}_0)+S_n(\widehat{g}_0)$ is a consistency estimator for $E\Vert\Gamma_{h,g}\Vert_q^2$.
Following the argument in the proof of Corollary 2 of \citet{szekely2007measuring}, we have $n\operatorname{AMGT}_n(\widehat{h}_0,\widehat{g}_0)^2/S_n\to Q$ in distribution, where $E(Q)=1$, and $Q$ is a nonnegative quadratic form of centered Gaussian random
variable which can be represented as $Q=\sum_{j=1}^\infty\lambda_{j}\chi_{1j}^2$.
\end{proof}

\subsection{Proof of Proposition~\ref{prop:6}}
\begin{proof}
If $\H_0$ is correct, we have established the asymptotic distribution of $T_n(\widehat{h}_0,\widehat{g}_0)$.
Proposition~\ref{prop:6}(a) follows from the result: if $Q$ is a quadratic form of
centered Gaussian random variables and $E(Q)=1$, then $\operatorname{pr}(Q\geq z^2_{1-\alpha/2})\leq\alpha$ for all $0<\alpha\leq0.215$ \citep{szekely2003extremal}.

If the generalized tetrad constraint \eqref{eq:6} does not hold, then $\operatorname{AMGT}(h_0,g_0)\neq0$ by Proposition~\ref{prop:4}(b).
Because $\operatorname{AMGT}_n(\widehat{h}_0,\widehat{g}_0)\to\operatorname{AMGT}(h_0,g_0)$ in probability by Theorem~\ref{thm:5}(a), we have $T_n(\widehat{h}_0,\widehat{g}_0)=n\operatorname{AMGT}_n(\widehat{h}_0,\widehat{g}_0)^2/S_n\to\infty$ in probability under Assumption~\ref{assum:6} and a parallel assumption for $S_n(\widehat{g}_0)$, and Proposition~\ref{prop:6}(b) follows.
\end{proof}

\section{Additional simulation results with observed covariates}\label{sec:simu-cov}

We further evaluate the performance of the tetrad tests in the presence of observed covariates.
Data are generated from the following model, 
where we need to adjust for a covariate $V$  for testing the null hypothesis,
\begin{equation}
\begin{gathered}
X=0.5+U+\alpha_2U^2+0.5V+\varepsilon_1,\ Y=-1+U+\beta_2U^2+V+\delta X+\varepsilon_2,\\
Z=0.5+U+V+\varepsilon_3,\ W=1+U+0.5V+\varepsilon_4,\\ (U,V,\varepsilon_1,\varepsilon_2,\varepsilon_3,\varepsilon_4)^{\T}\thicksim N(\boldsymbol{0},\boldsymbol{I}_6).
\end{gathered}
\end{equation}
Table~\ref{tab:s3} presents the   settings of  parameters we consider in simulations.
In Setting (I), the null hypothesis $\H^\prime_0$ is correct and the variables follow a linear  model.
In Setting (II), the null hypothesis $\H^\prime_0$ is correct with nonlinear effects of $U$ on $X$ and $Y$.
The values of $(\alpha_2,\beta_2)$ in (II.a) and (II.b) indicate the degree of nonlinearity. 
In Setting (III), $\H^\prime_0$ is incorrect and the model is linear.
The value of $\delta$  in (III.a) and (III.b)  indicates the degree of deviation from $\H^\prime_0$.

\begin{table}[tb]
\centering
\caption{Settings of parameters in simulations with observed covariates} 
{
\setlength{\aboverulesep}{0pt}
\setlength{\belowrulesep}{0pt}
\begin{tabular}{lllll}
\toprule
\multirow{3}{*}{$\H^\prime_0$ Correct} &(I)& & $\alpha_2=\beta_2=\delta=0$ & \\
\cmidrule{2-5}
&\multirow{2}{*}{(II)}& (a)& \multirow{2}{*}{$\delta=0$} &  $(\alpha_2,\beta_2)=(0.1,0.2)$\\
& & (b)&  &  $(\alpha_2,\beta_2)=(0.3,0.4)$\\
\midrule
\multirow{2}{*}{$\H^\prime_0$ Incorrect} &\multirow{2}{*}{(III)}& (a)& \multirow{2}{*}{$\alpha_2=\beta_2=0$} &    $\delta=0.3$\\
&&  (b)& &    $\delta=0.5$\\
\bottomrule
\end{tabular}
}
\label{tab:s3}
\end{table}

We apply both the classical tetrad test (CT) and  the proposed generalized tetrad test  for testing  $\H^\prime_0: X\ind Y\ind Z\ind W\mid (U,V)$ at significance level 0.05.
In the classical tetrad test, we obtain the test statistic based on  the two tetrad differences
$\sigma_{\widetilde{X}\widetilde{Y}}\sigma_{\widetilde{Z}\widetilde{W}}-\sigma_{\widetilde{X}\widetilde{W}}\sigma_{\widetilde{Y}\widetilde{Z}}=0$ and $\sigma_{\widetilde{X}\widetilde{Y}}\sigma_{\widetilde{Z}\widetilde{W}}-\sigma_{\widetilde{X}\widetilde{Z}}\sigma_{\widetilde{Y}\widetilde{W}}=0$, where $(\widetilde{X},\widetilde{Y},\widetilde{Z},\widetilde{W})$ are residuals after linear regression on $V$.
In the generalized tetrad test, we consider two different methods for estimation of the confounding bridge function, 
including  the GMM  and the PSMD method.
The corresponding tests are denoted by GT$_1$ and GT$_2$, respectively.
For GT$_1$, we set $h(w,v;\boldsymbol{\theta})=\theta_0+\theta_1 w+\theta_2 v$ and $\boldsymbol{B}(z,v)=(1,z,v)^{\T}$ for the GMM approach,
and for GT$_2$ we set we set $(\boldsymbol{\xi}^{q_n},\boldsymbol{p}^{k_n},\operatorname{Pen}(h),\lambda_n)$ to $(\operatorname{Pol}(4),\operatorname{Pol}(7),\Vert h\Vert_{L^2}^2+\Vert\nabla_w h\Vert_{L^2}^2+\Vert\nabla_v h\Vert_{L^2}^2,4\times10^{-5})$ for the PSMD approach,
where $\operatorname{Pol}(r)$ denotes bivariate polynomials of total degree at most $(r-1)$, and $\nabla_w h,\nabla_v h$ denote the partial derivatives of $h$ with respect to $w$ and $v$, respectively.
The estimation of $g_0(z,v)$ in GT$_1$ and GT$_2$ is analogous.
We implement the generalized tetrad test based on the test statistic $T_n(\widehat{h}_0,\widehat{g}_0)$.

For each setting, we replicate 1000 simulations at sample size 500 and 1000, respectively. 
Table~\ref{tab:s4} summarizes the power of the tests.
The results are analogous to those in Section~\ref{sec:simulation} without observed covariates.
In linear models, all tests can  control the empirical type I error when $\H^\prime_0$ is correct, and have power approaching unity as the sample size increases when $\H^\prime_0$ is incorrect;
the power of CT is higher than that of the generalized tetrad tests.
In nonlinear models, CT has large type I error when $\H^\prime_0$ is correct,
while the generalized tetrad test GT$_2$ can still control the type I error.
In practice, we recommend implementing and comparing multiple applicable methods  for  testing $\H^\prime_0$, 
both the classical and generalized tetrad tests, both parametric and nonparametric working models,
to obtain a reliable  conclusion.

\begin{table}[tb]
    \centering
    \caption{Power of tests at significance level 0.05.
Columns in gray  correspond to  sample size 1000  and otherwise for 500}
    {
\setlength{\aboverulesep}{0pt}
\setlength{\belowrulesep}{0pt}
    \begin{tabular}{ccc>{\columncolor{gray!40}}cc>{\columncolor{gray!40}}cc>{\columncolor{gray!40}}ccc}
\toprule
Settings &&\multicolumn{2}{c}{GT$_1$}& \multicolumn{2}{c}{GT$_2$} &\multicolumn{2}{c}{CT}  \\
\midrule
\multirow{3}{*}{$\H^\prime_0$ Correct} &(I)       &0.002 &0.000 &0.002 &0.002 &0.038 &0.050   \\ 
\cmidrule{2-8}
&(II.a)     &0.004 &0.011 &0.004 &0.002 &0.064 &0.086  \\
&(II.b)     &0.211 &0.705 &0.009 &0.004 &0.484 &0.852  \\
\midrule
\multirow{2}{*}{$\H^\prime_0$ Incorrect} &(III.a)    &0.399 &0.841 &0.529 &0.849 &0.743 &0.967  \\ 
&(III.b)    &0.855 &0.994 &0.878 &0.993 &0.954 &0.999  \\
\bottomrule
\addlinespace[10pt]
\end{tabular}
}
    \label{tab:s4}
\end{table}

\section{Additional real data analysis results}
\subsection{The results of the generalized tetrad tests for all permutations of the four variables in HS data}
\label{ssec:add-hs}
The testing results shown in this subsection are supplementary to Section~\ref{ssec:application-hs}. 
We conduct the generalized tetrad tests GT$_1$, GT$_2$, and GT$_3$ for different permutations of the four variables $(X,Y,Z,W)$, with the confounding bridge functions estimated in the same way as those in Section~\ref{ssec:application-hs}.
Table~\ref{tab:s6} shows the $p$-values of the tests for all permutations. 
In the sense of the two triples $\{(X,Y,Z),(X,Y,W)\}$, there are 12 different permutations of the four variables in total.

\begin{table}[tb]
   \centering
   \caption{The $p$-values of tests for every permutation of four variables in HS data.
We use $(i,j,k,l)$ to denote a permutation of the four variables. For example, $(1,2,3,4)$ stands for $(X,Y,Z,W)$ and $(1,3,2,4)$ for $(X,Z,Y,W)$}\label{tab:s6}
    {
\setlength{\aboverulesep}{0pt}
\setlength{\belowrulesep}{0pt}
    \begin{tabular}{cccc}
\toprule
Permutations & GT$_1$ & GT$_2$ &GT$_3$\\
\midrule
$(1,2,3,4)$ & 0.417  & 0.504  & 0.477 \\
$(1,3,2,4)$ & 0.451  & 0.505  & 0.471 \\
$(1,4,2,3)$ & 0.650  & 0.500  & 0.456 \\
$(2,1,3,4)$ & 0.363  & 0.471  & 0.477 \\
$(2,3,1,4)$ & 0.366  & 0.451  & 0.486 \\
$(2,4,1,3)$ & 0.286  & 0.629  & 0.424 \\
$(3,1,2,4)$ & 0.428  & 0.444  & 0.558 \\
$(3,2,1,4)$ & 0.373  & 0.454  & 0.485 \\
$(3,4,1,2)$ & 0.413  & 0.534  & 0.443 \\
$(4,1,2,3)$ & 0.624  & 0.539  & 0.411 \\
$(4,2,1,3)$ & 0.303  & 0.732  & 0.408 \\
$(4,3,1,2)$ & 0.474  & 0.561  & 0.438 \\
\bottomrule
\end{tabular}
 }   
\end{table}
At significance level 0.05, none of the testing results reject the null hypothesis, which reinforces that the observed-data distribution is consistent with the model that the four scores are mutually independent conditional on the covariates and an unobserved factor.

\subsection{The results of the generalized tetrad tests for all permutations of the four variables in the WVS data}
\label{ssec:add-wvs}
The testing results in this subsection are supplementary to Section~\ref{ssec:application-wvs}.
We conduct the generalized tetrad tests GT$_1$, GT$_2$, and GT$_3$ for different permutations of the four variables $(X,Y,Z,W)$, with the confounding bridge functions estimated in the same way as those in Section~\ref{ssec:application-wvs}. 
Table~\ref{tab:s7} shows the $p$-values of the tests for all permutations. 

\begin{table}[tb]
   \centering
   \caption{The $p$-values of tests for every permutation of four variables in the WVS data.
We use $(i,j,k,l)$ to denote a permutation of the four variables. For example, $(1,2,3,4)$ stands for $(X,Y,Z,W)$ and $(1,3,2,4)$ for $(X,Z,Y,W)$}\label{tab:s7}
    {
\setlength{\aboverulesep}{0pt}
\setlength{\belowrulesep}{0pt}
    \begin{tabular}{cccc}
\toprule
Permutations& GT$_1$ & GT$_2$ &GT$_3$\\
\midrule
$(1,2,3,4)$ & 0.235 &  0.281  & 0.275 \\
$(1,3,2,4)$ & 0.281 &  0.284  & 0.282 \\
$(1,4,2,3)$ & 0.228 &  0.282  & 0.256 \\
$(2,1,3,4)$ & 0.231 &  0.270  & 0.259 \\
$(2,3,1,4)$ & 0.235 &  0.235  & 0.249 \\
$(2,4,1,3)$ & 0.180 &  0.226  & 0.200 \\
$(3,1,2,4)$ & 0.271 &  0.308  & 0.280 \\
$(3,2,1,4)$ & 0.227 &  0.274  & 0.260 \\
$(3,4,1,2)$ & 0.241 &  0.293  & 0.256 \\
$(4,1,2,3)$ & 0.345 &  0.287  & 0.298 \\
$(4,2,1,3)$ & 0.261 &  0.245  & 0.244 \\
$(4,3,1,2)$ & 0.291 &  0.256  & 0.244 \\
\bottomrule
\end{tabular}
}
\end{table}
At significance level 0.05, none of the testing results reject the null hypothesis, which reinforces that the observed-data distribution is consistent with the model that the responses to the four questions are mutually independent conditional on the covariates and an unobserved factor.

\putbib
\end{bibunit}

\end{document}